\documentclass[12pt]{article}
\usepackage{verbatim}

\usepackage{thmtools, thm-restate}
\usepackage{tikz}

\usepackage{fullpage}
\usepackage{parskip}
\usepackage[utf8]{inputenc}
\usepackage[english]{babel}
\usepackage{mathpazo}

\usepackage{authblk}

\usepackage[hyphens]{url}
\usepackage[hypertexnames=false,colorlinks=true,linkcolor=blue,citecolor=blue,urlcolor=black]{hyperref}
\usepackage{breakurl}

\usepackage{amsmath, amsfonts, amssymb}
\usepackage{mathtools}
\usepackage{braket}

\usepackage{xspace}

\def\squareforqed{\hbox{\rlap{$\sqcap$}$\sqcup$}}
\def\qed{\ifmmode\squareforqed\else{\unskip\nobreak\hfil
    \penalty50\hskip1em\null\nobreak\hfil\squareforqed
    \parfillskip=0pt\finalhyphendemerits=0\endgraf}\fi}

\newtheorem{theorem}{Theorem}
\newtheorem{lemma}[theorem]{Lemma}
\newtheorem{fact}[theorem]{Fact}

\newenvironment{proof}{\begin{trivlist}\item[]{\flushleft\bf Proof }}
  {\qed\end{trivlist}}
  
\newcommand{\thmsp}{\vspace{2mm}}

\newcommand{\tspan}{\textup{span}}
\newcommand{\subz}[1]{\normalsize( #1 \normalsize)_z}

\newcommand{\op}[1]{\mathsf{#1}}
\newcommand{\opU}{\op{U}} 
\newcommand{\opA}{\op{A}} 
\newcommand{\opB}{\op{B}} 

\newcommand{\opF}{\op{F}} 
 
\newcommand{\opGt}{\op{\tilde{G}}} 
 
\newcommand{\opW}{\op{W}} 

\newcommand{\opid}{\op{I}}

\newcommand{\proj}{\Pi}
\newcommand{\proje}{\proj_{\op{e}}}
\newcommand{\projo}{\proj_{\op{o}}}

\newcommand{\hilb}{\mathcal{H}}
\newcommand{\nth}[1]{\ensuremath{{#1}^{\textup{th}}}}

\renewcommand\bra[1]{{\langle{#1}|}}
\renewcommand\ket[1]{{|{#1}\rangle}}
\newcommand{\ketbra}[2]{\ket{#1}\!\bra{#2}}
\newcommand{\inner}[2]{\langle{#1}|{#2}\rangle}
\newcommand{\kouter}[1]{|{#1}\rangle\!\langle{#1}|}

\newcommand{\marked}{g}
\newcommand{\markedt}{{\tilde{g}}}
\newcommand{\kmarked}{\ket{\marked}}
\newcommand{\kmarkedt}{\ket{\markedt}}

\newcommand{\selfloop}{\circlearrowleft}
\newcommand{\sel}{s}
\newcommand{\kselfloop}{\ket{\selfloop}}

\newcommand{\plus}{+}
\newcommand{\minus}{-}
\newcommand{\kplus}{\ket{\plus}}
\newcommand{\kminus}{\ket{\minus}}
\newcommand{\kminusp}{\kminus^{\perp}}

\newcommand{\eg}{{a_{00}}}
\newcommand{\keg}{\ket{\eg}}

\newcommand{\feval}{\lambda}
\newcommand{\rotf}{f}
\newcommand{\tfp}{{\rotf_+}}
\newcommand{\tfm}{{\rotf_-}}
\newcommand{\kfp}{\ket{\tfp}}
\newcommand{\kfpun}{\ket{\tfp^{un}}}
\newcommand{\kfm}{\ket{\tfm}}
\newcommand{\fangle}{\eta}
\newcommand{\opFo}{\opF_1}
\newcommand{\opFt}{\opF_2}
\newcommand{\fo}{{f_1}}
\newcommand{\ft}{{f_2}}
\newcommand{\kfo}{\ket{\fo}}
\newcommand{\kft}{\ket{\ft}}

\newcommand{\ustat}{{\opU_0}}
\newcommand{\kustat}{\ket{\ustat}}

\newcommand{\tpsi}{\psi}
\newcommand{\kpsi}{\ket{\tpsi}}
\newcommand{\kpi}{\ket{\pi}}
\newcommand{\piz}{{\pi_z}}
\newcommand{\kpiz}{\ket{\piz}}
\newcommand{\tevec}{\zeta}
\newcommand{\tevecb}{\overline{\tevec}}
\newcommand{\evec}{\ket{\tevec}}
\newcommand{\evecb}{\ket{\tevecb}}

\newcommand{\aco}{a}

\newcommand{\gamm}{{\gamma_-}}
\newcommand{\gams}{{\gamma_*}}

\newcommand{\phione}{{\varphi_1}}

\newcommand{\ncols}{{n_c}}
\newcommand{\nrows}{{n_r}}
\newcommand{\nsize}{N}

\newcommand{\kl}{{kl}}
\newcommand{\minusj}{{m_\kl}}
\newcommand{\tminusj}{{-_\kl}}
\newcommand{\kminusj}{\ket{\tminusj}}
\newcommand{\tpsij}{{\psi_\kl}}
\newcommand{\kpsij}{\ket{\tpsij}}

\newcommand{\repsi}{\Re(\kpsi)}
\newcommand{\impsi}{\Im(\kpsi)}
\newcommand{\rea}{\Re(\aco)}

\newcommand{\kp}{k'}

\newcommand{\klp}{{kl'}}
\newcommand{\kpl}{{k'l}}
\newcommand{\kplp}{{k'l'}}
\newcommand{\tk}{\tilde{k}}
\newcommand{\tl}{\tilde{l}}
\newcommand{\tkp}{\tilde{k}'}

\newcommand{\signk}{\eps_k}
\newcommand{\signl}{\eps_l}

\newcommand{\projwfo}{\proj_{(+1)}}

\newcommand{\ts}{s}
\newcommand{\ks}{\ket{\ts}}
\newcommand{\refs}{\op{S}}
\newcommand{\opT}{\op{T}}
\newcommand{\te}{{e_\alpha}}
\newcommand{\ke}{\ket{\te}}
\newcommand{\keperp}{\ket{\te^\perp}}

\newcommand{\kv}{\ket{v}}
\newcommand{\kvone}{\ket{{v_1}}}
\newcommand{\kvtwo}{\ket{{v_2}}}
\newcommand{\kvthree}{\ket{{v_3}}}
\newcommand{\kvim}{\ket{{v'}}}
\newcommand{\kminusjb}{\ket{\overline{\tminusj}}}

\newcommand{\cz}{\mathsf{cz}}
\newcommand{\eps}{\epsilon}
\newcommand{\sign}{\textup{sign}}

\newcommand{\smallspace}{\mskip 2mu plus 1mu minus 1mu\relax}
\newcommand{\tinyspace}{\mskip 1mu plus 0.3mu minus 0.5mu\relax}

\newcommand{\myeqref}[1]{equation~\eqref{#1}}

\title{Spatial Search via Memoryless Walk with Selfloop}

\author{Peter H{\o}yer}
\author{Janet Leahy} 
\affil{Department of Computer Science, University of Calgary, Canada}

\date{April 18, 2022}

\begin{document}
\maketitle


\begin{abstract}

The defining feature of memoryless quantum walks is that they operate on the vertex space of a graph, and therefore can be used to produce search algorithms with minimal memory. We present a memoryless walk that can find a unique marked vertex on a two-dimensional grid. Our walk is based on the construction proposed by Falk, which tessellates the grid with squares of size $2 \times 2$. Our walk uses minimal memory, $O(\! \sqrt{\nsize \log \nsize})$ applications of the walk operator, and outputs the marked vertex with vanishing error probability. To accomplish this, we apply a selfloop to the marked vertex---a technique we adapt from interpolated walks. We prove that with our explicit choice of selfloop weight, this forces the action of the walk asymptotically into a single rotational space. We characterize this space and as a result, show that our memoryless walk produces the marked vertex with a success probability asymptotically approaching~one.
\end{abstract}

\section{Introduction}

Search problems are one of the foundational applications of quantum algorithms, and are one of the situations in which quantum algorithms are proven to provide a speedup over their classical counterparts. 
For example, Grover's search algorithm~\cite{Gro96} can find a single element in an $N$-element database with $\Theta(\sqrt{N})$ quantum queries, while any classical algorithm requires $\Omega(N)$ queries for the same task.

In spatial search, the goal is to find a ``marked'' element on a graph, with the restriction that in a single time step, amplitude can only be moved between adjacent vertices on the graph.  Such restrictions can emerge from some underlying physical structure~\cite{AA05}, or from the computational cost of moving from one vertex to another~\cite{Amb07}.  Quantum spatial search was first considered by Benioff~\cite{Ben02}, who showed that direct application of Grover's search on the grid does not yield a speedup over classical algorithms. A~near-quadratic speed-up was then discovered using a divide-and-conquer quantum algorithm~\cite{AA05} and quantum walks~\cite{AKR05,CG04,San08}.

Quantum walks are the quantum counterparts of classical random walks. Memoryless, or coinless, walks are a less commonly studied type of quantum walk, but have several advantages over other models. As the name suggests, they operate directly on the vertex space of the graph. This differentiates them from coined or Szegedy-style quantum walks, which use extra registers to encode the previous location of the walker. Memoryless walks can therefore be used to produce search algorithms with optimal memory requirements. They also have a simple structure, alternating two or more non-commuting reflections that can be derived from tessellations of the underlying graph. A proposal for implementing memoryless walks using superconducting microwave resonators is given by~\cite{MOP17}, and an implementation of memoryless walks on IBM quantum computers is in \cite{AAMP20}.

Staggered walks are a class of memoryless walks and were defined by Portugal et al.~in~\cite{PSFG16}. Staggered walks are based on graph tessellations into cliques, which enforces the spatial search constraint and gives a method for constructing memoryless walks on general graphs. Ref.~\cite{PSFG16} shows that any quantum walk on a graph $G$ in the standard Szegedy model~\cite{Sze04} can be converted to a staggered walk on the underlying line graph of~$G$. The conversion preserves the asymptotic cost, the success probability, and the space requirement. Staggered walks have also been used to derive relationships between memoryless, coined, Szegedy, and continuous quantum walks~\cite{PBF15, Por16, Por16b, POM17, CP18, KPSS18}.

A memoryless walk on the line was given by Patel, Raghunathan, and Rungta in~\cite{PRR05}, who give and analyse a walk with Hamiltonians, and they note that their walk operator resembles the staggered fermion formalism. In~\cite{PRR10}, Patel, Raghunathan and Rahaman present numerical simulations on the extension of this walk to the $\nsize$-vertex grid. Their results show that it finds a unique marked vertex in $\Theta(\sqrt{\nsize \log \nsize})$ applications of the walk operator with success probability $\Theta(\frac{1}{\log \nsize})$, and that by using an ancilla qubit, the success probability can be improved to $\Theta(1)$.

Falk~\cite{Fal13} gives a memoryless walk on the two-dimensional grid, constructing a discrete walk operator by reflecting about two alternating tessellations of the grid. The walk by Falk was analysed by Ambainis, Portugal and Nahimovs~\cite{APN15}, who proved that it finds a unique marked vertex on the grid using $\Theta(\sqrt{\nsize \log \nsize})$ applications of the walk operator with success probability~$\Theta(\frac{1}{\log \nsize})$. Portugal and Fernandes~\cite{PF17} give a memoryless walk with Hamiltonians that also finds a unique marked vertex on the two-dimensional grid with $\Theta(\sqrt{\nsize \log \nsize})$ applications of the walk operator and success probability $\Theta(\frac{1}{\log \nsize})$.

The success probability of \cite{PRR05,PRR10}, \cite{Fal13} and~\cite{PF17} is sub-constant. The success probability can be improved to $\Theta(1)$ by applying amplitude amplification~\cite{BHMT02}, but that would increase the number of steps by a factor of $\Theta(\sqrt{\log \nsize})$ and introduce additional operators in an implementation. Two alternative methods~\cite{APN15,PF17} for increasing the success probability are the post-processing local neighborhood search in~\cite{ABNOR12} and Tulsi's proposal of adding an ancilla qubit to the grid~\cite{PRR10, Tul08}.

In this paper, we give a memoryless walk that finds the marked vertex in $\Theta(\sqrt{\nsize \log \nsize})$ steps with vanishing error probability. Our walk uses minimal memory and preserves the simple structure given by alternating tessellations of the grid. To do this, we augment the basic memoryless construction by showing how the interpolated walks of Krovi et al.~\cite{KMOR16} can be adapted to the vertex space of a graph. We define our memoryless walk using the tessellations shown in Figure \ref{fig:tessellation}, which divide the grid into squares of size $2 \times 2$. This is the same tessellation structure analysed by~\cite{APN15}, who proved that using Falk's construction, the maximum success probability of the corresponding memoryless walk scales with $\Theta(\frac{1}{\log \nsize})$. By using a selfloop to force the action of the walk into a single two-dimensional subspace, we modify their walk so that the maximum success probability asymptotically approaches~one.

Classically, the interpolated version of a random walk is constructed by adding weighted selfloop edges to marked vertices. Applying Szegedy's isometry~\cite{Sze04} to interpolated walks produces quantum walks that can find a unique marked vertex on any graph~\cite{KMOR16}. This approach has recently been used to solve the spatial search problem on any graph with a quadratic speedup over classical random walks, even for the case of multiple marked vertices. Solutions of this form have been obtained for both the discrete~\cite{AGJK20, AGJ21} and continuous~\cite{ACNR21arxiv} models using quantum walks that operate on the edge space of a graph. Quantum interpolated walks can be simulated by controlled quantum walks~\cite{DH17b}, preserving the asymptotic cost, the success probability, and the space requirement.

To extend the approach of~\cite{KMOR16} to the memoryless setting, we introduce a new state corresponding to a selfloop on the marked vertex. Rather than reflecting about the marked vertex, we reflect about an interpolation between the selfloop state and the marked vertex. The result is a memoryless walk parametrized by the weight of the selfloop,~$\sel$. We prove in our analysis that with our explicit choice of weight, this forces the evolution of the initial state into a single rotational subspace of our walk operator. As a result, our walk achieves a success probability $1 - O(\frac{1}{\log \nsize})$ with $\Theta(\sqrt{\nsize \log \nsize})$ steps and minimal memory.

Most of our analysis considers the eigenvector of an operator with the smallest positive eigenphase. We use the term ``slowest eigenvector'' to refer to this eigenvector, and ``slowest rotational subspace'' to refer to the subspace spanned by this eigenvector and its conjugate.

To prove our main result, we present a set of techniques for analysing memoryless walks. We analyse an operator $\opW$, whose spectra are completely known, composed with a two-dimensional rotation~$\opF$. As part of our proof, we determine the asymptotic behaviour of both the smallest positive eigenphase of $\opW \opF$ and its associated eigenvector. The result is a precise asymptotic description of the slowest rotational subspace of our walk operator.

The techniques we use to obtain this description are general enough that they could be applied in other contexts as well.

We begin by defining our walk in Section~\ref{sec:walk_defs}. An overview of the paper layout and proof structure is given in Section~\ref{sec:proof_strategy}.

\section{Walk construction and main result}
\label{sec:walk_defs}

We consider the task of finding a unique marked vertex on a two-dimensional grid with $\nrows$ rows and $\ncols$ columns, where both $\nrows$ and $\ncols$ are even. The grid boundaries are those of a torus, so there are edges between vertices $(i, \ncols-1)$ and $(i, 0)$ for $0 \leq i < \nrows$, and between vertices $(\nrows - 1, j)$ and $(0, j)$ for $0 \leq j < \ncols$. The total number of vertices is given by $\nsize = \nrows \times \ncols$. After Lemma \ref{lem:ustat}, we restrict to the case where $\nrows = \ncols = \sqrt{\nsize}$.

We construct a memoryless quantum walk with a selfloop. Our walk operates on a space of dimension $\nsize + 1$, which is optimal for this task. The vertex at position $(i,j)$ is represented by the quantum state $\ket{i,j}$. We introduce a new state $\kselfloop$ corresponding to a selfloop on the marked vertex. Our approach and terminology are inspired by the interpolated walks of~\cite{KMOR16}, which add selfloop edges to marked vertices in Szegedy-style walks.

Our walk applies two alternating reflections about the faces of the graph. Here, we follow the construction used in~\cite{APN15}. For $0 \leq i < \frac{\nrows}{2}$, $0 \leq j < \frac{\ncols}{2}$, define
\begin{align*}
\ket{a_{ij}} &= \frac{1}{2} \sum_{i',j' = 0}^1 \ket{2i + i', 2j + j'}, \\
\ket{b_{ij}} &= \frac{1}{2} \sum_{i',j' = 0}^1 \ket{2i + 1 + i', 2j + 1 + j'}.
\end{align*}

The sets $\{\ket{a_{ij}}\}_{i,j}$ and $\{\ket{b_{ij}}\}_{i,j}$ each specify a partition of the grid into $2 \times 2$ squares, positioned at even and odd indices, respectively. These sets form the tessellations depicted in Figure \ref{fig:tessellation}.

Define projections onto the even and odd partitions as\begin{alignat*}{2}
\proje = \sum_{i=0}^{\frac{\nrows}{2}-1} \sum_{j=0}^{\frac{\ncols}{2}-1} \ketbra{a_{ij}}{a_{ij}},
\hspace{2cm} & \projo = \sum_{i=0}^{\frac{\nrows}{2}-1} \sum_{j=0}^{\frac{\ncols}{2}-1} \ketbra{b_{ij}}{b_{ij}},
\end{alignat*}
and let \begin{align}
\opA &= 2\proje + 2 \ketbra{\selfloop}{\selfloop} - \opid, \\
\opB &= 2\projo + 2 \ketbra{\selfloop}{\selfloop} - \opid.
\end{align}

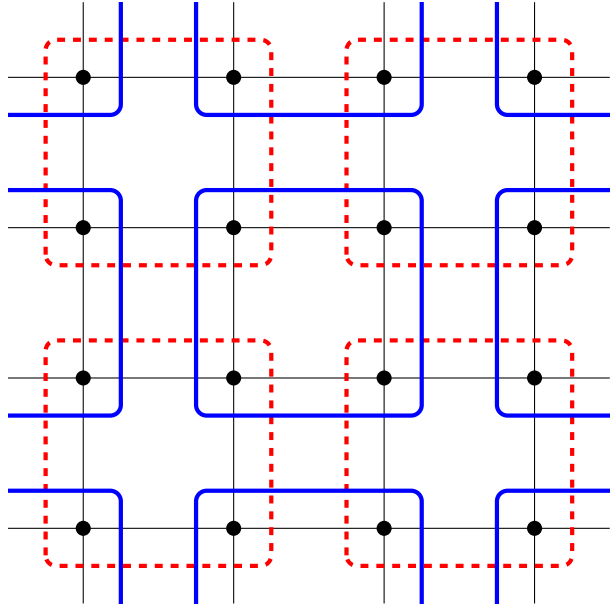
\begin{figure} 
\centering
\begin{tikzpicture}[scale=.5]
\clip (-2,-2) rectangle (14, 14);
\draw[step=4] (-2,-2) grid (14,14);

\foreach \x in {0,4,...,12} {
    \foreach \y in {0,4,...,12} {
        \fill[color=black] (\x,\y) circle (0.2);
    }
}
\foreach \ax in {-1, 7} {
	\foreach \ay in {-1, 7} {
		\draw[rounded corners, red, ultra thick, dashed] (\ax, \ay) rectangle (\ax + 6,\ay + 6);	
	}
}
\foreach \ax in {-5, 3, 11} {
	\foreach \ay in {-5, 3, 11} {
		\draw[rounded corners, blue, ultra thick] (\ax, \ay) rectangle (\ax + 6,\ay + 6);	
	}
}
\end{tikzpicture}
\caption{Staggered tessellations of the grid into squares of size $2 \times 2$. The construction based on alternating reflections about this tessellation structure was originally proposed by~\cite{Fal13}, who ran numerical simulations using squares of size $4 \times 4$.}
\label{fig:tessellation}
\end{figure}

In~\cite{APN15}, these reflections are alternated with a reflection about the marked state.
This is the standard way of adding finding behaviour to memoryless walks, where an operator composed of two or more non-commuting reflections is alternated with a reflection of the marked vertices.
In our walk, we replace the reflection of the marked vertex with a reflection of an interpolated state with parameter $0 \leq \sel \leq 1$. Letting $\kmarked$ denote the marked vertex, we define the interpolated state to be $\kmarkedt = \sqrt{\sel}\kmarked + \sqrt{1-\sel} \kselfloop$. Our input-dependent reflection is then \begin{align}
\opGt = \opid - 2\ketbra{\markedt}{\markedt}.
\end{align}
By selecting an appropriate value for $\sel$, we are able to force the action of the walk asymptotically into a single two-dimensional subspace. We set $\sel = 1 - \frac{1}{\nsize + 1}$, which is close to the value $\sel = 1 - \frac{1}{\nsize}$ used by~\cite{KMOR16} for interpolated walks.

Given a fixed value for $\sel$, we define a single step of our walk to be the operator \begin{align}
\opU &= \opB \opGt \opA \opGt.
\end{align}

We apply the walk to the initial state $\kpi$, which is defined by $\inner{i,j}{\pi} = \frac{1}{\sqrt{\nsize}}$ for all vertices $(i,j)$, and with $\inner{\selfloop}{\pi} = 0$. This allows us to present our main result.

\thmsp

\begin{restatable}[Main result]{theorem}{mainresult}
\label{thm:main_result}
Fix $\sel = 1 - \frac{1}{\nsize + 1}$ and suppose $\nrows = \ncols$. Then there exists a constant $c > 0$ such that after $c \sqrt{\nsize \log \nsize}$ applications of $\opU$ to $\kpi$, measuring the state will produce $\kselfloop$ with probability $1 - e(\nsize)$, where $e(\nsize) \in O(\frac{1}{\log \nsize})$.
\end{restatable}

\thmsp

We remark that given the state $\kselfloop$, one can obtain the marked state $\kmarked$ through amplitude amplification~\cite{BHMT02}. This can be done by alternating the reflection $\opGt$ with a reflection about either $\kmarked$ or $\kselfloop$. After $\lfloor \frac{\pi}{2} (\arcsin(\frac{1}{\sqrt{\nsize + 1}}))^{-1} \rfloor \in \Theta(\sqrt{\nsize})$ steps of amplification, measuring the resulting state will produce $\kmarked$ with probability $1 - e(\nsize)$, where $e(\nsize) \in O(\frac{1}{\nsize})$. Note that both the error probability and the query complexity for obtaining $\kmarked$ from $\kselfloop$ are dominated by the cost of finding the state $\kselfloop$ as in Theorem~\ref{thm:main_result}.

\subsection{Proof strategy} \label{sec:proof_strategy}

Our proof of Theorem~\ref{thm:main_result} is based on the analysis of an intermediate walk operator, which we define below.

To simplify the analysis, we also introduce a change of basis. This will allow us to compute necessary properties of the walk spectra, as the eigenvectors of $\opB \opA$ factor into product states under this basis change. Let $\cz$ be an operator acting on the pair of least significant bits in a tensor product space, with the action defined by $\ket{i,j} \mapsto - \ket{i,j}$ if both $i$ and $j$ are even, and being the identity otherwise. Let $\cz$ act trivially on $\kselfloop$. For any operator $\op{X}$, let $\op{X}_{z} = \cz \circ \op{X} \circ \cz$, and let $\kpiz = \cz \kpi$.

We define \begin{align}
\opW &= \subz{\opB \opA}, \\
\opF &= \subz{\opA \opGt \opA \opGt},
\end{align}
so that
\begin{equation}
  \opU = \big(\opB \opA \big) \big(\opA \opGt \opA \opGt \big)
  = \subz{\opW \opF}. \label{eq:U_decomp}
\end{equation}

This is analogous to the decomposition of $\opU$ used by~\cite{APN15} in their analysis.
Note that $\opW$ is input-independent, while $\opF$ depends on both $\sel$ and the marked vertex.
We show in Section~\ref{sec:F_decomp} that $\opF$ is a two-dimensional rotation, and therefore can be decomposed as the product of two reflections, which we write as $\opF = \opFo \opFt$.

Our intermediate walk consists of $\opW$ composed with only $\opFo$.
This choice of intermediate walk has the advantage that it consists of a real operator whose spectra are completely known, composed with a one-dimensional reflection. This allows us to apply existing results about operators of this type, including the eigenvector analysis of~\cite{Amb07} and the flip-flop theorem from~\cite{DH17b}. An overview of these results is given in Appendix~\ref{app:flipflop}.

Our proof is based on a tight characterization of the slowest rotational subspace of the intermediate walk $\opW \opFo$. This characterization is developed in Section~\ref{sec:WFo}, where we prove asymptotic properties of the rotational angle and the spanning vectors. To our knowledge, this is the first case of the flip-flop theorem being used to derive properties of a subspace in this way.
We show that with our choice of selfloop, the action of $\opW \opFo$ on $\kpiz$ can be reduced asymptotically to a Grover-like rotation in the slowest two-dimensional rotational subspace of $\opW \opFo$. This key property is what allows our main algorithm to achieve a success probability asymptotically close to 1.

To prove our main result, we relate the slowest rotational subspaces of $\opW \opFo$ and $\opW \opF$ in Section~\ref{sec:U}. We show that composing $\opW \opFo$ with the reflection $\opFt$ does not significantly alter the slowest rotational subspace, and therefore that $\opW \opF$ has the same asymptotic behaviour as $\opW \opFo$ when applied to $\kpiz$. The proof of Theorem \ref{thm:main_result} follows from the basis-change relationship between $\opU$ and $\opW \opF$ given in 
\myeqref{eq:U_decomp}.

With our approach, we are able to derive precise statements about the behaviour of an operator composed with a two-dimensional rotation. This addresses a more general challenge in analysing quantum algorithms, and may have applications outside of memoryless walks.

\section{Decomposition of $\opF$} \label{sec:F_decomp}

In this section, we derive the exact form of the rotation $\opF$. We show that it can be decomposed into two one-dimensional reflections, $\opFo$ and $\opFt$, which we compose sequentially with $\opW$ in Sections~\ref{sec:WFo} and \ref{sec:U}. 

Without loss of generality, assume $\kmarked = \ket{0,0}$ is the marked vertex. Consider the three-dimensional subspace spanned by $\kmarked$, $\kselfloop$ and $\keg$, the even-indexed square containing $\ket{0,0}$. The operator $\opF$ only acts non-trivially in this subspace, which is spanned by $\kselfloop$ and the two orthonormal states
\begin{align}
  \kplus  &= \frac{1}{\sqrt{3}} (\kmarked + \keg), \label{eq:kplus_def}\\
  \kminus &= \hphantom{\frac{1}{\sqrt{3}}} (\kmarked - \keg). \label{eq:kminus_def}
\end{align}

\thmsp
\begin{lemma} 
\label{lem:F_evecs}
Let $0 \leq \fangle \leq \frac{\pi}{3}$ be such that 
\begin{equation}
  \sin^2 (\fangle) = \frac{3}{4}\sel.
\end{equation}
Set $\feval = e^{\imath 4 \fangle}$. Then
$\opF = \opid + (\feval - 1) \kouter{\tfp} + (\feval^{-1} - 1)
\kouter{\tfm}$, where the two non-trivial eigenvectors are
\begin{align}
  \kfp
  &= \frac{1}{\sqrt{2}} \Bigg[
    \kplus
    - \imath
    \frac{1}{\sqrt{4-3\sel}}\Big(
    \sqrt{\sel} \kminus
    - 
    2 \sqrt{1-\sel} \kselfloop
    \Big)
    \Bigg], \\
  \kfm
  &= \frac{1}{\sqrt{2}} \Bigg[
    \kplus
    + \imath
    \frac{1}{\sqrt{4-3\sel}}\Big(
    \sqrt{\sel} \kminus
    - 
    2 \sqrt{1-\sel} \kselfloop
    \Big)
    \Bigg].
\end{align}
\end{lemma}

\begin{proof}
Observe that $\opF$ is a real-valued operator that only acts non-trivially on the span of $\kplus$, $\kminus$ and $\kselfloop$. Therefore, any complex eigenvalues of $\opF$ must come in conjugate pairs, and $\opF$ can have at most three non-trivial eigenvectors. Using the observation that $\subz{\opA \opGt}$ is real-valued, any $(-1)$-eigenspace of $\opF$ would necessarily have even dimension. Thus, $\opF$ must have either two or zero non-trivial eigenvectors.

Define $\kfpun = \sqrt{2(4-3\sel)}\kfp$ and compute 
\begin{align*}
\subz{\opA \opGt} \kfpun &= \bigg[ -(2- 3\sel)\frac{\sqrt{4-3\sel}}{2} - \imath \sqrt{3\sel} (4-3\sel) \bigg] \kplus  \\
&+ \bigg[ -\frac{\sqrt{4-3\sel}}{2} \sqrt{3}\sel + \imath \sqrt{\sel}(2- 3\sel) \bigg]\kminus \nonumber \\
&+ \bigg[ \sqrt{3\sel(1-\sel)(4-3\sel)} - \imath \sqrt{1-\sel}(2- 3\sel) \bigg] \kselfloop \nonumber \\
&= -\frac{1}{2} \Big( (2-3\sel) + \imath \sqrt{3\sel} \sqrt{4-3\sel} \Big) \kfpun \\[2mm]
&= - e^{\imath 2 \fangle} \kfpun.
\end{align*}
This shows that $\kfp$ is an eigenvector of $\opF = \subz{\opA \opGt} \subz{\opA \opGt}$ with eigenvalue $(-e^{\imath 2 \fangle})^2 = \feval$. It follows that the entrywise conjugate of $\kfp$, given by $\kfm$, must be an eigenvector of $\opF$ with eigenvalue $\feval^{-1}$.
\end{proof}

Lemma \ref{lem:F_evecs} shows that $\opF$ is a rotation by $4 \fangle$ of a single two-dimensional space, spanned by $\kfp$ and $\kfm$. Therefore, $\opF$ can be decomposed as the product of two one-dimensional reflections. We choose these reflections as follows.

\thmsp
\begin{fact} \label{lem:F_decomp}
Define \begin{align}
\kfo &= \frac{1}{\sqrt{4 - 3\sel}} \Big( \sqrt{\sel} \kminus - 2 \sqrt{1 - \sel} \kselfloop \Big), \\
\kft &= \sin(2\fangle) \kplus + \cos(2 \fangle) \kfo,
\end{align}
and let \begin{align}
\opFo &= \opid - 2\kouter{\fo},\\
\opFt &= \opid - 2\kouter{\ft}.
\end{align}
Then $\opF = \opFo \opFt$.
\end{fact}

\section{Structure of $\opW$} \label{sec:W}

We give exact formulas for the eigenvectors and eigenphases of $\opW$, as well as a decomposition of the $(\nsize +1)$-dimensional domain into subspaces that are invariant under $\opW$. These properties are required for the precise characterisation of $\opW \opFo$ and $\opW \opF$ in later sections.

\subsection{Spectra of $\opW$}

Recall that $\kselfloop$ is trivially a $(+1)$-eigenvector of $\opW$. The remaining $\nsize$ eigenvectors of $\opW$ can be indexed by $k$ and $l$ as follows.

For $0 \leq k < \nrows/2$, let $\tk = \frac{2 \pi k}{\nrows}$, and for
$0 \leq l < \ncols/2$, let $\tl = \frac{2 \pi l}{\ncols}$. Let
$\signk = \sign(\cos \tk)$ and $\signl = \sign(\cos \tl)$. Define the sign of zero to be $+1$. This case occurs when
$\nrows$ or $\ncols$ is divisible by 4, since $\cos(\tk) = 0$ when
$k=\nrows/4$ and $\cos(\tl) = 0$ when $l=\ncols/4$.

Define
\begin{align}
	p_\kl &= \sqrt{1 - \cos^2 \tk \cos^2 \tl},  \\
	\theta_\kl &= \signk \signl \,\textup{acos}(1 - 2 p^2_\kl), \label{eq:theta_def}
\end{align}
and \begin{align*}
    r_\kl^{\pm} &= \sqrt{2 \bigg(1 \pm \frac{\sin \tk \cos \tl }{p_\kl}\bigg)}
    = \sqrt{1+ \frac{\sin \tl}{p_\kl}} 
      \pm \signl \sqrt{1- \frac{\sin \tl}{p_\kl}}, \\[4.8mm]
    c_\kl^{\pm} &= \sqrt{2 \bigg(1 \pm \frac{\cos\tk \sin \tl}{p_\kl}\bigg)}
    = \sqrt{1+ \frac{\sin \tk}{p_\kl}} 
      \pm \signk \sqrt{1- \frac{\sin \tk}{p_\kl}}.
\end{align*}

If $k=l=0$, then $p_{00}= 0$ and the division by zero is
ill-defined. In this case, we define
\begin{equation*}
  r_{00}^+ = r_{00}^- = c_{00}^+ = c_{00}^- = \sqrt{2}.
\end{equation*}

For each $0 \leq k < \nrows/2$ and $0 \leq l < \ncols/2$, there is an
eigenvector $\ket{w_\kl}$ of $\opW$ with eigenvalue $e^{\imath\theta_\kl}$. This $\ket{w_{\kl}}$ is the product state
\begin{equation}
  \ket{w_\kl} = \ket{u_\kl} \otimes \ket{v_\kl},
\end{equation}
where the factors are given by the normalized states

\begin{minipage}[c]{.50\textwidth}
  \begin{align*}
    \ket{u_\kl} &= \sqrt{2} \ket{\phi^k_r} \circ (\ket{1_{\nrows/2}} \otimes \ket{r_\kl}), \\[1.5mm]
    \ket{v_\kl} &= \sqrt{2} \ket{\phi^l_c} \circ (\ket{1_{\ncols/2}} \otimes \ket{c_\kl}), 
  \end{align*}
\end{minipage}
\begin{minipage}[c]{.50\textwidth}
  \begin{align*}
    \ket{r_\kl} &= \frac{1}{2}
                  \begin{bmatrix}
                    r_\kl^{-} \\ r_\kl^{+}
                  \end{bmatrix}, \\
    \ket{c_\kl} &= \frac{1}{2}
                  \begin{bmatrix}
                    c_\kl^{-} \\ c_\kl^{+}
                  \end{bmatrix}. 
  \end{align*}
\end{minipage}
\newline

Here, $\circ$ denotes entry-wise multiplication and $\ket{1_n}$ is the all-ones vector of dimension $n$.  
The Fourier states are
$\ket{\phi^k_r} = \frac{1}{\sqrt{\nrows}} \sum_{i=0}^{\nrows-1} \omega^{ik}_{\nrows} \ket{i}$ and $\ket{\phi^k_c} = \frac{1}{\sqrt{\ncols}} \sum_{i=0}^{\ncols-1} \omega^{ik}_{\ncols} \ket{i}$, where $\omega_n = e^{2\pi \imath/n}$ denotes
the \nth{n} root of unity.

Let 
\begin{align*}
  \ket{r^1_\kl} &= \op{XZ} \ket{r_\kl},\\
  \ket{c^1_\kl} &= \op{XZ} \ket{c_\kl}.
\end{align*}
Then replacing $\ket{r_\kl}$ and $\ket{c_\kl}$ with $\ket{r^1_\kl}$ and
$\ket{c^1_\kl}$ in the above construction yields an eigenvector
with eigenphase $-\theta_\kl$, and replacing either one of the two yields an
eigenvector with eigenvalue~$1$. Let $\ket{u_\kl^1}$ and $\ket{v_\kl^1}$ be
defined with $\ket{r^1_\kl}$ and $\ket{c^1_\kl}$, respectively. We denote $\ket{u^0_\kl} = \ket{u_\kl}$, $\ket{v^0_\kl} = \ket{v_\kl}$. Let
$\ket{w_\kl^B}$ be defined accordingly for $B \in \{00, 01, 10, 11\}$.

By this definition, the $\nsize$ eigenvectors $\{\ket{w_\kl^B}\}$ of $\opW$ constitute an orthonormal basis for the grid, where the eigenvector $\ket{w_\kl^B}$ has eigenphase $\theta_\kl, 0, 0, -\theta_\kl$ for $B = 00, 01, 10, 11$. 

\thmsp
\begin{fact}
Both $\kpi$ and $\kpiz$ are $(+1)$-eigenvectors of $\opW$.
\end{fact}
\begin{proof}
Based on \myeqref{eq:theta_def}, we note $\theta_{00} = 0$. The fact follows from the observation that \begin{gather*}
\kpi = \ket{w_{00}^{00}},\\
\kpiz = \frac{1}{2} \Big( \ket{w_{00}^{00}} + \ket{w_{00}^{01}} + \ket{w_{00}^{10}} - \ket{w_{00}^{11}} \Big).
\end{gather*}
\end{proof}

\subsection{Invariant subspaces} \label{sec:invariant_subspaces}

We partition the domain of $\opW$ into subspaces $\opW_{kl}$. The number of subspaces depends on the parity of $\frac{\ncols}{2}$. If $\frac{\ncols}{2}$ is odd, there are $\frac{(\nrows + 2) (\ncols - 2)}{8} + 1 $ invariant subspaces, and if $\frac{\ncols}{2}$ is even, there are $ \frac{(\nrows + 2) \ncols}{8} + \lfloor \frac{\nrows}{4} \rfloor - 1 $ invariant subspaces.
 Each $\opW_{kl}$ is spanned by a set of eigenvectors of $\opW$ with eigenphase $\theta_{kl}$, and is therefore invariant under the action of $\opW$. The projection onto subspace $\opW_{kl}$ is denoted $\proj_{kl}$.

First, observe that $\theta_{00} = 0$, so the eigenvectors $\ket{w_{00}^{00}}$ and $\ket{w_{00}^{11}}$ are both $(+1)$-eigenvectors.
The $(+1)$-eigenspace of $\opW$ therefore has dimension $\frac{N}{2} + 2 + 1$, where the last dimension comes from the selfloop state. We denote the $(+1)$-eigenspace as $\opW_{00}$ and its associated projection as~$\proj_{00}$.

For any $0 \leq k < \nrows/2$ and $0 \leq l < \ncols/2$, define

\begin{minipage}{.48\textwidth}
  \begin{equation*}
    k' = \begin{cases}
      \frac{\nrows}{2} - k & \textup{ if $0< k < \frac{\nrows}{2}$ } \\
      0 & \textup{ if $k = 0$}
    \end{cases}
  \end{equation*}
\end{minipage}
\begin{minipage}{.1\textwidth}
and
\end{minipage}
\begin{minipage}{.4\textwidth}
  \begin{equation*}
    l' = \begin{cases}
      \frac{\ncols}{2} - l & \textup{ if $0< l < \frac{\ncols}{2}$ } \\
      0 & \textup{ if $l = 0$.}
    \end{cases}
  \end{equation*}
\end{minipage}

Note that $k=k'$ exactly when $k=0$ or $k= \frac{\nrows}{4}$, and similarly for $l = l'$.

For $0 < k < \frac{\nrows}{2}, k \neq \frac{\nrows}{4}$ and $0 < l < \frac{\ncols}{4}$, define 
\begin{equation*}
  \opW_{kl} =
  \tspan \Big \{ \ket{w^{00}_\kl}, \ket{w^{00}_\kplp}, \ket{w^{11}_\kpl}, \ket{w^{11}_\klp} \Big \}.
\end{equation*}
Each of these subspaces has an associated eigenphase $\theta_\kl \neq 0, \pi$. When $\nrows$ or $\ncols$ is a multiple of 4, $\opW$ also has a $(-1)$-eigenspace. In this case, there are subspaces with eigenphase $\pi$ given by
 \begin{equation*}
  \opW_{kl} =
  \tspan \Big \{ \ket{w^{00}_\kl}, \ket{w^{00}_\kplp}, \ket{w^{11}_\kpl}, \ket{w^{11}_\klp} \Big \}.
\end{equation*}
for $k = \frac{\nrows}{4}$, $0 < l < \frac{\ncols}{4}$ and $l = \frac{\ncols}{4}$, $0 < k \leq \frac{\nrows}{4}$. These four vectors will be distinct unless both $k = \frac{\nrows}{4}$ and $l = \frac{\ncols}{4}$, in which case $\dim(\opW_{kl}) = 2$.

For $k = 0$ and $0 < l < \frac{\ncols}{2}$, the corresponding invariant subspace is \begin{align*}
  \opW_{kl} =
  \tspan \Big \{ \ket{w^{00}_{\kl}}, \ket{w^{11}_\klp} \Big \},
\end{align*}
and similarly for $l = 0$ and $0 < k < \frac{\nrows}{2}$, \begin{align*}
  \opW_{kl} =
  \tspan \Big \{ \ket{w^{00}_{\kl}}, \ket{w^{11}_\kpl} \Big \}.
\end{align*}

Partitioning the domain in this way allows us to closely analyse the behaviour of $\kplus$ and $\kminus$ under the action of $\opW$.

\thmsp
\begin{restatable}{lemma}{claimone}
\label{lem:claim1}
The following statements hold. \begin{align}
    \proj_{kl} \kplus  \label{eq:kp_perp_km}
    &\perp \proj_{kl} \kminus
    &\textup{ for all subspaces $\opW_{kl}$} \\[5pt]
    \|\proj_{kl} \kplus\|^2	\label{eq:kplusj}
    &= \frac{2}{3} \,\frac{\dim(\opW_{kl})}{N} 
    &\textup{ for all $k,l$ not both $0$} \\
    \|\proj_{kl} \kminus\|^2	\label{eq:kminusj}
    &= {2} \,\frac{\dim(\opW_{kl})}{N} 
    &\textup{ for all $k,l$ not both $0$}\\
    \|\proj_{00} \kplus\|^2	\label{eq:kpluszero}
    &= \frac{2}{3} \frac{N+2}{N} \\
    \|\proj_{00} \kminus	\|^2	\label{eq:kminuszero}
    &= \frac{4}{N}.
\end{align}
\end{restatable}

\begin{proof}
See Appendix \ref{app:proof_of_claim1}.
\end{proof}

\section{Reduction to the slowest subspace} \label{sec:WFo}

Our memoryless walk, given in Theorem~\ref{thm:main_result}, achieves a success probability asymptotically close to 1. This is possible because our choice of $\sel$ reduces the walk asymptotically to a rotation in a single two-dimensional subspace. As we show, this subspace is exactly the slowest rotational subspace of the applied walk operator.
We prove this by giving a tight description of the smallest positive eigenvalue of the walk operator and its associated eigenvector. We also show that the two-dimensional rotation induced by the walk maps the initial state to the desired state, $\kselfloop$.

We consider two walk operators. The first consists of the real operator $\opW$ composed with the one-dimensional reflection $\opFo$. The second consists of $\opW$ composed with the two-dimensional rotation $\opF$. We discuss the first operator in this section, and then use the results to derive properties of the second operator in Section~\ref{sec:U}.

First, we prove three key lemmas about the slowest rotational subspace of $\opW \opFo$. Lemma \ref{lem:WFo_ephase} gives a tight bound on the smallest positive eigenphase of $\opW \opFo$, and Lemmas \ref{lem:aco} and \ref{lem:kpsi} characterize the asymptotic behaviour of its associated eigenvector. These three lemmas show that the action of $\opW \opFo$ on $\kpiz$ can be reduced to a rotation in the slowest rotational subspace. These results are the basis of the methods used in Section~\ref{sec:U} to prove our main result.

The intermediate operator $\opW \opFo$ is fundamentally a tool for analysis. However, it is interesting to note that $\opW \opFo$ can be used directly to find the marked state. The proof of the corollary below follows from a similar argument to our proof of the main theorem.

\thmsp
\begin{restatable}{corollary}{WFothm}
\label{thm:intermediate}
Fix $\sel = 1 - \frac{1}{\nsize + 1}$ and suppose $\nrows = \ncols$. Then there exists a constant $c > 0$ such that after $c \sqrt{\nsize \log \nsize}$ applications of $\opW \opFo$ to $\kpiz$, measuring the state will produce $\kselfloop$ with probability $1 - e(\nsize)$, where $e(\nsize) \in O(\frac{1}{\log \nsize})$.
\end{restatable}
\thmsp

Comparing Corollary \ref{thm:intermediate} with Theorem \ref{thm:main_result} shows that most of the finding behaviour of the memoryless walk comes from the first reflection $\opFo$. In Section~\ref{sec:U}, we show that composing $\opW \opFo$ with the second reflection $\opFt$ only changes the behaviour of the walk slightly, leading to our proof of Theorem \ref{thm:main_result}.

Our proofs of both Corollary \ref{thm:intermediate} and our main theorem rely on the following observation.

\thmsp
\begin{lemma} \label{lem:ustat}
The (unnormalized) vector \begin{align}
\kustat = \kpiz - \sqrt{\frac{\sel}{(1-\sel)N}} \kselfloop
\end{align}
is a $(+1)$-eigenvector for each of $\opW$, $\opFo$, and $\opFt$.
\end{lemma}

\begin{proof}
Both $\kpiz$ and $\kselfloop$ are $(+1)$-eigenvectors of $\opW$, so $\kustat$ is a $(+1)$-eigenvector of $\opW$. To show $\kustat$ is a $(+1)$-eigenvector of $\opFo$, we compute \begin{align*}
\inner{\ustat}{\fo} = \frac{2 \sqrt{\sel}}{\sqrt{N (4 - 3\sel)}} + \frac{\sqrt{\sel}}{\sqrt{4 - 3\sel}} \inner{\piz}{\minus} = 0.
\end{align*}
Finally, $\kplus$ is orthogonal to both $\kpiz$ and $\kselfloop$, so $\kustat$ is orthogonal to $\kft$. Therefore, $\kustat$ is also a $(+1)$-eigenvector of $\opFt$.
\end{proof}

For the rest of the paper, we assume a square grid, so $\nrows = \ncols = \sqrt{\nsize}$. We also fix $\sel = 1 - \frac{1}{N+1}$. This choice of selfloop weight means $\kustat = \kpiz - \kselfloop$, so we can decompose the initial state $\kpiz$ as \begin{align}
\kpiz = \frac{1}{2} \kustat + \frac{1}{2} \Big( \kpiz + \kselfloop \Big). \label{eq:piz_decomp}
\end{align}

We show in Section~\ref{sec:WFo_evec} that for our chosen $\sel$, $\kpiz + \kselfloop$ lies asymptotically in the slowest rotational subspace of $\opW \opFo$. Therefore, $\opW \opFo$ can be used to apply a negative phase to this portion of the initial state. This rotates the state $\kpiz$ to the state $\frac{1}{2} \kustat - \frac{1}{2} \big( \kpiz + \kselfloop \big) = -\kselfloop$, as we make precise in our proof of Corollary \ref{thm:intermediate}.

When $\sel = 1 - \frac{1}{N+1}$, note that the vector $\kfo$ has the form
\begin{align}
\kfo &= \sqrt{\frac{N}{N+4}} \kminus - \frac{2}{\sqrt{N+4}} \kselfloop. \label{eq:kfo_decomp}
\end{align}

\subsection{Smallest eigenphase of $\opW \opFo$} \label{sec:WFo_angle}

We choose the operator $\opW \opFo$ as our intermediate step in the analysis of $\opW \opF$ because it is the composition of a well-characterized real operator with a one-dimensional reflection. This allows us to apply results from the literature about operators of this type. Here, we show how these results can be used to obtain a tight bound on the smallest positive eigenphase of $\opW \opFo$.

An overview of the applied results is given in Appendix~\ref{app:flipflop}. 

\thmsp
\begin{lemma} \label{lem:WFo_ephase}
The smallest positive eigenphase $\phione$ of $\opW \opFo$ satisfies $\phione \in \Theta ( \frac{1}{\sqrt{\nsize \log \nsize}} )$.
\end{lemma}

\begin{proof}
First, we note that the smallest positive eigenvalue of $\opW$ is given by \begin{align}
\theta_{10} = \,\textup{acos} \bigg( 2 \cos^2 \! \bigg( \frac{2\pi}{\sqrt{\nsize}} \bigg) - 1 \bigg) = \frac{4\pi}{\sqrt{\nsize}}.
\end{align}
We know by Lemma \ref{lem:claim1} and the decomposition in \myeqref{eq:kfo_decomp} that $\kfo$ satisfies \begin{align*}
\| \proj_{00} \kfo \|^2 &= \frac{8}{N+4},  \\[5pt]
\| \proj_\kl \kfo \|^2 &= \frac{2 \dim(\opW_\kl)}{N+4}     &\textup{ for all $k,l$ not both $0$}. 
\end{align*}
In particular, $\kfo$ overlaps all eigenspaces of $\opW$, so by Theorem \ref{lem:flipflop_interlace}, $\phione < \theta_{10} = \frac{4\pi}{\sqrt{\nsize}}$.

By Lemma \ref{lem:flipflop_constraints}, $\phione$ must also satisfy the equation
\begin{align} \label{eq:flipflop_phione}
\| \proj_{00} \kfo \|^2 \cot \! \bigg( \frac{\phione}{2} \bigg) + \sum_{\kl \neq 00} \| \proj_\kl \kfo \|^2  \cot \! \bigg( \frac{\phione - \theta_\kl}{2} \bigg) = 0,
\end{align}
where the sum is taken over all invariant subspaces of $\opW$ except the $(+1)$-eigenspace, $\opW_{00}$.
Note that the $(-1)$-eigenspace is included this sum.
Because $\phione > 0$ and $\phione \in o(1)$, we have
\begin{align*} 
\| \proj_{00} \kfo \|^2 \cot \! \bigg( \frac{\phione}{2} \bigg) \in \Theta \Big( \frac{1}{\phione \nsize} \Big).
\end{align*}
Therefore, for \myeqref{eq:flipflop_phione} to hold, it must be the case that \begin{align} \label{eq:phione_sum}
\sum_{\kl \neq 00} \| \proj_\kl \kfo \|^2 \cot \! \bigg(\frac{\theta_\kl - \phione}{2}\bigg) \in \Theta \Big( \frac{1}{\phione \nsize} \Big).
\end{align}

We argue that there cannot be a solution $\phione \in \Theta(\frac{1}{\sqrt{\nsize}})$. By Fact \ref{lem:sum0}, we know that for such a $\phione$, the sum in \myeqref{eq:phione_sum} has order $\Omega(\frac{\log \nsize}{\sqrt{\nsize}})$. Therefore, it must be the case that $\phione \in o(\frac{1}{\sqrt{\nsize}})$.

By Fact \ref{lem:sum0}, we also know that if $\phione \in o(\frac{1}{\sqrt{\nsize}})$, then the sum in \myeqref{eq:phione_sum} has order $\Theta(\phione \log \nsize)$. Due to the requirement $\phione \log \nsize \in \Theta(\frac{1}{\phione \nsize})$, the only possible solution is $\phione \in \Theta ( \frac{1}{\sqrt{\nsize \log \nsize}} )$, as stated.
\end{proof}

\subsection{Slowest eigenvector of $\opW \opFo$} \label{sec:WFo_evec}

In this section, we use a constraint-solving approach to analyse the eigenvector of $\opW \opFo$ associated with eigenphase $\phione$. By determining its asymptotic behaviour, we show that the $\kpiz + \kselfloop$ component of \myeqref{eq:piz_decomp} lies in the span of this eigenvector and its conjugate. This shows that $\opW \opFo$ can be applied to rotate the initial state $\kpiz$ to the target state $\kselfloop$.

Define $\kminusp = \kminus + \frac{2}{\sqrt{N}} \kpiz$
to be the (unnormalized) component of $\kminus$ that is orthogonal to $\kpiz$. Note that by Lemma \ref{lem:claim1}, this vector is orthogonal to $(+1)$-eigenspace of $\opW$.

Let $\evec$ be an unnormalized eigenvector of $\opW \opFo$ with eigenphase $\alpha \neq 0, \pi$ and $\inner{\tevec}{\piz} \neq 0$, scaled such that $\inner{\tevec}{\piz} = \frac{1}{2}$. Because $\evec$ is perpendicular to $\kustat$, this implies $\inner{\tevec}{\selfloop} = \frac{1}{2}$. We decompose $\evec$ as 
\begin{align}
\evec =  \aco \kminusp + \frac{1}{2} \kpiz + \frac{1}{2} \kselfloop + \kpsi, \label{eq:evec_def}
\end{align}
where $\kpsi$ is an unnormalized vector orthogonal to $\kminusp$. By analysing the asymptotic behaviour of $\aco$ and $\kpsi$, we show that the real part of $\evec$ tends to $\frac{1}{2} (\kpiz + \kselfloop)$ when $\alpha = \phione$.

Note that any eigenvector of $\opW$ with eigenphase $\theta_\kl$ that is orthogonal to $\kfo$ is also an eigenvector of $\opW \opFo$ with eigenphase $\theta_\kl$. Therefore, $\proj_\kl \kpsi$ is some scalar multiple of $\proj_\kl \kminusp$ for each $k,l$. We determine this scalar factor in Lemma~\ref{lem:WFo_constraints}.

We further decompose both $\kminusp$ and $\kpsi$ into the invariant subspaces of $\opW$. Both $\kminusp$ and $\kpsi$ are orthogonal to the $(+1)$-eigenspace $\opW_{00}$, so we write the decomposition as
\begin{align}
\kminusp &= \sum_{kl \neq 00} \minusj \kminusj, \\
\kpsi &= \sum_{kl \neq 00} \kpsij,
\end{align} 
where the vectors $\kminusj$ are normalized for all $k,l$. We know by Lemma \ref{lem:claim1} that $\minusj = \sqrt{\frac{2 \dim(\opW_{kl})}{N}}$. The vectors $\kpsij$ in the decomposition of $\kpsi$ are unnormalized.

\thmsp
\begin{lemma} \label{lem:WFo_constraints}
The following equations must be satisfied.
\begin{align} \label{eq:WFo_constraint1}
\frac{8 \aco (N-4)}{\sqrt{N}(N+4)} - \frac{16}{N+4} = e^{\imath \alpha} - 1
\end{align}
\begin{align} \label{eq:WFo_constraint2}
\inner{\tminusj}{\tpsij} = \minusj \Bigg[ \aco - \frac{\sqrt{N}}{4} (e^{\imath \alpha} - 1) \bigg( \frac{1}{1 - e^{\imath (\alpha - \theta_\kl)}} \bigg) \Bigg] \quad \textup{ for all $k,l$ not both $0$}.
\end{align}
\end{lemma}

\begin{proof}
By definition, $\evec$ is an eigenvector of $\opW \opFo$ with eigenphase $\alpha$. We obtain the lemma by expanding the equation $\opW \opFo \evec = e^{\imath \alpha} \evec$ and solving for constraints.

Observe that
\begin{align*}
\inner{\fo}{\tevec} = \frac{1}{\sqrt{N+4}} \bigg( \frac{\aco (N-4)}{\sqrt{N}} - 2 \bigg).
\end{align*}
Using this property, we compute \begin{align*}
\opW \opFo \evec &= \opW \evec - 2 \opW \inner{\fo}{\tevec} \kfo \\
&= \gamm \opW \kminusp + \gams (\kpiz + \kselfloop) + \opW \kpsi,
\end{align*}
where \begin{align*}
\gamm &= \aco - \frac{2}{N+4} \Big( \aco (N-4) - 2 \sqrt{N} \Big), \\
\gams &= \frac{1}{2} +  \frac{4}{\sqrt{N}(N+4)} \Big( \aco (N-4) - 2 \sqrt{N} \Big).
\end{align*}

Setting $\opW \opFo \evec = e^{\imath \alpha} \evec$ and comparing coefficients on $\kpiz$, we get \begin{align*}
\gams = \frac{1}{2} e^{\imath \alpha},
\end{align*}
which can be expanded to give \myeqref{eq:WFo_constraint1}.

To get \myeqref{eq:WFo_constraint2}, we first solve $\opW \opFo \evec = e^{\imath \alpha} \evec$ on the subspace $\opW_\kl$ to get \begin{align}
\inner{\tminusj}{\tpsij} = \minusj \bigg( \frac{\aco e^{\imath \alpha} - \gamm e^{\imath \theta_\kl}}{e^{\imath \theta_\kl} - e^{\imath \alpha}} \bigg). \label{eq:kminus_kpsi_rot}
\end{align}
Next, we use \myeqref{eq:WFo_constraint1} to rewrite $\gamm$ as \begin{align*}
\gamm = \aco - \frac{\sqrt{N}}{4} (e^{\imath \alpha} - 1).
\end{align*}
Substituting this expression for $\gamm$ into \myeqref{eq:kminus_kpsi_rot} produces \myeqref{eq:WFo_constraint2}.
\end{proof}

Now, fix $\evec$ to be the eigenvector with eigenphase $\phione$. By Lemma \ref{lem:flipflop_constraints}, we know that $\inner{\tevec}{\piz} \neq 0$, so the constraints given in Lemma~\ref{lem:WFo_constraints} apply. These constraints, together with the bound on $\phione$ from Lemma \ref{lem:WFo_ephase}, define asymptotic bounds on $\aco$ and the real part of $\kpsi$. We use this to show that as $\nsize$ increases, the real part of $\evec$ converges to $\frac{1}{2} \kpiz + \kselfloop$.
 
Let $\evecb$ denote the entrywise conjugate of $\evec$. Then $\evec$ and $\evecb$ span the slowest rotational subspace of $\opW \opFo$. In this way, the following lemmas provide a close description of the spanning eigenvectors for the slowest rotational subspace of $\opW \opFo$.

\thmsp
\begin{lemma} \label{lem:aco}
Let $\alpha = \phione$. Then $|\aco| \in \Theta ( \frac{1}{\sqrt{\log N}} )$.
\end{lemma}

\begin{proof}
By Lemma \ref{lem:WFo_ephase}, we know that $|e^{\imath \phione} - 1| \in \Theta ( \frac{1}{\sqrt{N \log N}} )$. Applying this to \myeqref{eq:WFo_constraint1} produces the stated bound.
\end{proof}

\thmsp
\begin{lemma} \label{lem:kpsi}
Let $\alpha = \phione$. Let $\kpsi = \repsi + \imath \impsi$, where both $\repsi$ and $\impsi$ are vectors with real entries. Then $\| \repsi \| \in O(\frac{1}{\sqrt{ \log N}})$.
\end{lemma}

\begin{proof}
From \myeqref{eq:WFo_constraint2}, we know that \begin{align*}
\kpsi &= \sum_{\kl \neq 00} \minusj \Bigg[ \aco - \frac{\sqrt{N}}{4} (e^{\imath \alpha} - 1) \bigg( \frac{1}{1 - e^{\imath (\alpha - \theta_\kl)}} \bigg) \Bigg] \kminusj \\
&= \aco \kminusp + \rho \kv, \label{eq:decomppsi}
\end{align*}
where we define \begin{align*}
\rho &= - \frac{\sqrt{N}}{4} (e^{\imath \alpha} - 1), \\
\kv &= \sum_{\kl \neq 00} \minusj \bigg( \frac{1}{1 - e^{\imath( \alpha - \theta_\kl)}} \bigg) \kminusj.
\end{align*}

We know from Lemma \ref{lem:aco} that $\|\ \aco \kminusp \| \in \Theta(\frac{1}{\sqrt{\log N}})$. Thus, it remains to consider $\rho \kv$.

Examining $\rho$ shows that for $\alpha \in \Theta(\frac{1}{\sqrt{ N \log N}})$, \begin{align*}
&|\Re(\rho)| \in \Theta \bigg( \frac{1}{\sqrt{N} \log N} \bigg), & |\Im(\rho)| \in \Theta \bigg( \frac{1}{\sqrt{\log N}} \bigg).
\end{align*}
Therefore, we can prove the lemma by showing that $\| \Re(\kv) \| \in O(\sqrt{\nsize \log \nsize})$ and $\| \Im(\kv) \| \in O(1)$.
Observe that for all $k,l$,

\begin{minipage}{.5\textwidth}
  \begin{align*}
  \Re \bigg( \frac{1}{1 - e^{\imath (\alpha - \theta_\kl)}} \bigg) = \frac{1}{2},
  \end{align*}
\end{minipage}
\begin{minipage}{.5\textwidth}
  \begin{align*}
  \Im \bigg( \frac{1}{1 - e^{\imath (\alpha - \theta_\kl)}} \bigg) = -\frac{1}{2} \cot \! \bigg( \frac{\theta_\kl - \alpha}{2} \bigg).
  \end{align*}
\end{minipage}

Using this property, we split the coefficients of $\kv$ into their real and imaginary parts, giving
\begin{align*}
\kv &= \frac{1}{2} \sum_{\kl \neq 00} \minusj \kminusj - \frac{\imath}{2} \sum_{\kl \neq 00} \minusj \cot \! \bigg( \frac{\theta_\kl - \alpha}{2} \bigg) \kminusj \\
&= \frac{1}{2} \kminusp - \frac{\imath}{2} \kvim, 
\end{align*}
where \begin{align*}
\kvim = \sum_{\kl \neq 00} \minusj \cot \! \bigg( \frac{\theta_\kl - \alpha}{2} \bigg) \kminusj.
\end{align*}
Note that $\kminusp$ is real-valued and has norm $\Theta(1)$.

We now bound the real and imaginary parts of $\kvim$.
Recall that for each subspace $\opW_\kl$ with eigenphase $0 < \theta_\kl < \pi$, there is a corresponding subspace with eigenphase $-\theta_\kl$, which we denote  $\overline{\opW_\kl}$. Because $\kminusp$ is real-valued, it must be the case that the normalized projection of $\kminusp$ onto $\overline{\opW_\kl}$ is $\kminusjb$, the entrywise conjugate of $\kminusj$. Using this property, we decompose $\kvim$ as 
\begin{align*}
\kvim &= \sum_{0 < \theta_\kl < \pi} \minusj \Bigg[ \cot \! \bigg( \frac{\theta_\kl - \alpha}{2} \bigg) \kminusj - \cot \! \bigg(\frac{\theta_\kl + \alpha}{2} \bigg) \kminusjb \Bigg] \notag \\
& \quad + \sum_{\theta_\kl = \pi} \minusj \cot \! \bigg(\frac{\theta_\kl - \alpha}{2} \bigg) \kminusj \\
&= \kvone + \kvtwo + \kvthree,
\end{align*}
where \begin{align*}
\kvone &= \sum_{0 < \theta_\kl < \pi} \minusj \Bigg[ \cot \! \bigg(\frac{\theta_\kl + \alpha}{2} \bigg) \kminusj - \cot \! \bigg(\frac{\theta_\kl + \alpha}{2} \bigg) \kminusjb \Bigg], \\
\kvtwo &= \sum_{0 < \theta_\kl < \pi} \minusj \Bigg[ \cot \! \bigg(\frac{\theta_\kl - \alpha}{2} \bigg) - \cot \! \bigg(\frac{\theta_\kl + \alpha}{2} \bigg) \Bigg] \kminusj, \\
\kvthree &= \sum_{\theta_\kl = \pi} \minusj \cot \! \bigg(\frac{\theta_\kl - \alpha}{2} \bigg) \kminusj.
\end{align*}
Note that the sums are taken over the invariant subspaces of $\opW$ whose eigenphases lie in the indicated range.
We bound the norms of these three components individually. First, observe that \begin{align*}
\kvthree =  \tan \! \bigg(\frac{\alpha}{2} \bigg) \! \sum_{\theta_\kl = \pi} \minusj \kminusj,
\end{align*}
where $\sum_{\theta_\kl = \pi} \minusj \kminusj$ is the projection of $\kminus$ onto the $(-1)$-eigenspace of $\opW$. Therefore, $\kvthree$ must be entirely real-valued, with norm $\| \kvthree \| \in O(\frac{1}{\sqrt{N \log N}})$ by Lemma \ref{lem:WFo_ephase}.

The vector $\kvone$ is entirely imaginary-valued, with norm \begin{align*}
\| \kvone \|^2 &= \sum_{0 < \theta_\kl < \pi} \minusj^2 \Bigg[ \cot \! \bigg(\frac{\theta_\kl + \alpha}{2} \bigg)^2 + \cot \! \bigg(\frac{\theta_\kl + \alpha}{2} \bigg)^2 \Bigg]   \\
&\leq \frac{16}{\nsize} \sum_{0 < \theta_\kl < \pi} \cot \! \bigg(\frac{\theta_\kl + \alpha}{2} \bigg)^2 \\
&= \frac{16}{\nsize} \sum_{0 < \theta_\kl < \pi} \! \Bigg( \frac{\cot(\frac{\theta_\kl}{2}) \cot(\frac{\alpha}{2}) - 1}{\cot(\frac{\alpha}{2}) + \cot(\frac{\theta_\kl}{2})} \Bigg)^2 \\
&\leq \frac{16}{\nsize} \cot^2 \! \bigg( \frac{\alpha}{2} \bigg) \sum_{0 < \theta_\kl < \pi} \! \Bigg( \frac{\cot(\frac{\theta_\kl}{2})}{\cot(\frac{\alpha}{2}) + \cot(\frac{\theta_\kl}{2})} \Bigg)^2.
\end{align*}
There are $O(\nsize)$ terms in the final sum, each of which is at most $1$, so $\| \kvone \|^2 \in O(\nsize \log \nsize)$.

Finally, the vector $\kvtwo$ has both real and imaginary parts, and has norm
\begin{align*}
\| \kvtwo \|^2 &= \sum_{0 < \theta_\kl < \pi} \minusj^2 \Bigg[ \cot \! \bigg(\frac{\theta_\kl + \alpha}{2} \bigg) - \cot \! \bigg(\frac{\theta_\kl - \alpha}{2} \bigg) \Bigg]^2 \\
&= \frac{2}{\nsize} \sum_{0 < \theta_\kl < \pi} \dim(\opW_\kl) \Bigg[ \cot \! \bigg(\frac{\theta_\kl + \alpha}{2} \bigg) - \cot \! \bigg(\frac{\theta_\kl - \alpha}{2} \bigg) \Bigg]^2 \\
&= \frac{1}{\nsize} \sum_{\kl \neq 00} \dim(\opW_\kl) \Bigg[ \cot \! \bigg(\frac{\theta_\kl + \alpha}{2} \bigg) - \cot \! \bigg(\frac{\theta_\kl - \alpha}{2} \bigg) \Bigg]^2.
\end{align*}
By Fact \ref{lem:sum4}, this implies $\| \kvtwo \|^2 \in O(\frac{1}{\log \nsize})$.

Combining the bounds on $\kvone$, $\kvtwo$ and $\kvthree$, we bound the norm of the real and imaginary parts of $\kvim$. Thus, \begin{align*}
&\|\Re(\kvim) \| \leq \| \kvthree \| + \| \kvtwo \| \in O \Big( \frac{1}{\sqrt{\log \nsize}} \Big), \\[1mm]
&\|\Im(\kvim) \| \leq \| \kvone \| + \| \kvtwo \| \in O(\sqrt{\nsize \log \nsize}).
\end{align*}

Because $\kv = \frac{1}{2} \kminusp - \frac{\imath}{2} \kvim$, this shows in particular that $\| \Re(\kv) \| \in O(\sqrt{\nsize \log \nsize})$ and $\| \Im(\kv) \| \in O(1)$. We combine this with the bounds on $\rho$ to obtain  $\| \repsi \| \in O(\frac{1}{\sqrt{ \log \nsize}})$ as stated.
\end{proof}

Lemmas \ref{lem:aco} and \ref{lem:kpsi} show that the real part of $\evec$ tends to $\frac{1}{2}(\kpiz + \kselfloop)$ as $\nsize$ increases. This implies that $\kpiz + \kselfloop$ lies asymptotically in the slowest rotational subspace of $\opW \opFo$. We make this precise in the following lemma.

\thmsp
\begin{lemma} \label{lem:WFo_proj}
Let $\proj_\phione$  denote the projection onto the slowest rotational subspace of $\opW \opFo$, which is spanned by the eigenvectors with eigenphases $\pm \phione$. Then \begin{align}
\Big \| \proj_\phione \Big( \kpiz + \kselfloop \Big) \Big \| = \sqrt{2} - O \Big( \frac{1}{\log N} \Big).
\end{align}
\end{lemma}

\begin{proof}
Let $\rea$ denote the real part of $\aco$. Observe that \begin{align*}
\evec + \evecb = \kpiz + \kselfloop + 2 \rea \kminusp + 2 \repsi.
\end{align*}
By Lemma \ref{lem:aco}, we have $|\rea| \in O(\frac{1}{\sqrt{\log N}})$, and by Lemma \ref{lem:kpsi} we have $\| \repsi \| \in O(\frac{1}{\sqrt{\log N}})$. Therefore, \begin{align*}
\| \evec + \evecb \| = \sqrt{2} + O \Big( \frac{1}{\log N} \Big).
\end{align*}
We can also see that \begin{align*}
\Big( \bra{\tevec} + \bra{\tevecb} \Big) \Big( \kpiz + \kselfloop \Big) = 2.
\end{align*}

Noting that $\proj_\phione$ denotes the projection onto $\tspan \{ \evec, \evecb \}$, this implies that \begin{align*}
\Big \| \proj_\phione \Big( \kpiz + \kselfloop \Big)  \Big \| \geq \sqrt{2} - O\Big( \frac{1}{\log N} \Big).
\end{align*}
\end{proof}

\subsection{Proof of corollary}

Applying the operator $\opW \opFo$ to the initial state $\kpiz$ yields an optimal algorithm for finding $\kselfloop$. The proof of this corollary follows from a similar argument to the proof of Theorem \ref{thm:main_result}. We apply our characterization of the slowest rotational subspace of $\opW \opFo$, given by Lemmas \ref{lem:WFo_ephase} and \ref{lem:WFo_proj}, to show the action of $\opW \opFo$ on $\kpiz$ is asymptotically restricted to this single subspace. The result is a Grover-like algorithm that rotates $\kpiz$ to our desired state $\kselfloop$.

\thmsp
\WFothm*

\begin{proof}
Recall the decomposition of $\kpiz$ in \myeqref{eq:piz_decomp}.
By Lemma \ref{lem:ustat}, we know that $\kustat$ is a $(+1)$-eigenvector of $\opW \opFo$. Letting $\proj_\phione$ denote the projection onto the slowest rotational subspace of $\opW \opFo$, we decompose $\kpiz + \kselfloop$ as \begin{align*}
\kpiz + \kselfloop = \proj_\phione \Big(\kpiz + \kselfloop \Big) + \ket{\bot}.
\end{align*}
for some vector $\ket{\bot}$. By Lemma \ref{lem:WFo_proj}, we know that $\| \ket{\bot} \| \in O(\frac{1}{\log N})$. We also know from Lemma \ref{lem:WFo_ephase} that the slowest rotational subspace has eigenphase $\phione \in \Theta(\frac{1}{\sqrt{N \log N}})$. Therefore, there exists a constant $c$ such that $c \sqrt{N \log N} = \lfloor \frac{\pi}{\phione} \rfloor = k$. After $k$ applications of $\opW \opFo$ to $\kpiz$, we get the state \begin{align*}
(\opW \opFo)^k \kpiz &= \frac{1}{2} \Big( \kpiz - \kselfloop \Big) - \frac{1}{2} \proj_\phione \Big( \kpiz + \kselfloop \Big) + \ket{\rho} \\
&= - \kselfloop + \frac{1}{2} \ket{\bot} + \ket{\rho}.
\end{align*}
Here, $\ket{\rho}$ is some state that captures both the result of applying $(\opW \opFo)^k$ to $\ket{\bot}$ and the small error incurred by the rounding of $\frac{\pi}{\phione}$, and has norm $\| \ket{\rho} \| \in O(\frac{1}{\log \nsize})$.
Thus, measuring the state will produce $\kselfloop$ with probability $1 - e(N)$, where $e(N) \in O(\frac{1}{\log N})$ as stated.
\end{proof}

\section{Finding with a memoryless walk} \label{sec:U}

We now present the proof of our main theorem, which is stated as follows.

\thmsp
\mainresult*
\thmsp

Our proof uses the decomposition of $\kpiz$ given in \myeqref{eq:piz_decomp}. We show that the state $\frac{1}{2} (\kpiz + \kselfloop)$ lies asymptotically in the slowest rotational subspace of $\opW \opF$. Therefore, $\opW \opF$ can be used to rotate $\kpiz$ to a state close to $-\kselfloop$. Applying the change of basis $\cz$ yields the result as stated.

The proof is based on relating the slowest rotational subspaces of $\opW \opFo$ and $\opW \opF = (\opW \opFo) \opFt$. We continue to apply the results from Appendix \ref{app:flipflop}, this time to analyse the real operator $\opW \opFo$ composed with the one-dimensional reflection $\opFt$. We show the slowest rotational subspaces of $\opW \opFo$ and $\opW \opF$ have the same asymptotic bound on the rotational angle, and that both asymptotically contain $\frac{1}{2} (\kpiz + \kselfloop)$. By using $\opW \opFo$ as an intermediate operator, we are thus able to tightly characterize the slowest rotational subspace of a real operator $\opW$, composed with a two-dimensional rotation $\opF$.

\subsection{Relationship with $\opW \opFo$}

We begin by relating the slowest rotational subspaces of $\opW \opF$ and $\opW \opFo$. Note that when $\sel = 1 - \frac{1}{\nsize + 1}$, the vector $\kft$ has the form \begin{align}
\kft &= \frac{\sqrt{3N} \sqrt{N+4} }{2(N+1)}  \kplus + \frac{2 - N}{2(N+1)} \kfo. \label{eq:kft}
\end{align}

Let the eigenphases of $\opW \opFo$ different from $0, \pi$ be denoted by $\pm \varphi_k$ for $1 \leq k \leq m$, where $0 < |\varphi_1| \leq |\varphi_2| \leq \cdots \leq |\varphi_m| < \pi$. Let the associated eigenvectors be $\ket{A_{k}^{\pm}}$. Both $\kft$ and $\opW \opFo$ are real-valued, so we can decompose $\kft$ into the eigenbasis of $\opW \opFo$ as \begin{align}
\kft = g_0 \ket{A_0} + \sum_{k=1}^{m} g_k \Big( \ket{A_k^+} + \ket{A_k^-} \Big) + g_{-1} \ket{A_{-1}}, \label{eq:kft_decomp}
\end{align}
where $\ket{A_0}$ is a $(+1)$-eigenvector, $\ket{A_{-1}}$ is a $(-1)$-eigenvector, and all $g_i$ are non-negative real numbers.

By Lemma \ref{lem:claim1}, $\proj_\kl \kplus \perp \proj_\kl \kfo$ for each invariant subspace $\opW_\kl$ of $\opW$ (including $\opW_{00}$). Therefore, the eigenvectors $\proj_\kl \kplus$ of $\opW$ are also eigenvectors of $\opW \opFo$ with the same eigenphases~$\theta_\kl$.

\thmsp
\begin{lemma} \label{lem:kft_coefficients}
The decomposition of $\kft$ in \myeqref{eq:kft_decomp} satisfies: \begin{align}
&g_0^2 = \frac{1}{2} + O \Big( \! \frac{1}{\sqrt{\nsize}} \Big)\\
&g_1^2 \in O \Big(\frac{1}{\log \nsize} \Big).
\end{align}
\end{lemma}

\begin{proof}
Recall that $0 < \phione < |\theta_\kl|$ for all nonzero eigenphases $\theta_\kl$ of $\opW$. Using the property that $\proj_\kl \kplus$ is an eigenvector of $\opW \opFo$ with eigenvalue $\theta_\kl$, this implies that $\bra{A_1^+} \proj_\kl \kplus = 0$ for all $k, l$, so $\inner{A_1^+}{\plus} = 0$.
Let $\projwfo$ denote the projection onto the $(+1)$-eigenspace of $\opW\opFo$. Then we have $\|\projwfo \kplus \|^2 = \| \proj_{00} \kplus\|^2 = \frac{2(N+2)}{3N}$ by Lemma \ref{lem:claim1}.

We bound $\| \projwfo \kfo \|^2$ by observing that 
\begin{align*}
\| \projwfo\kfo \|^2 &= \sum_{\theta_\kl = \pi} \| \proj_\kl \kfo \|^2 \\
&= \frac{N}{N+4} \sum_{\theta_\kl = \pi} \| \proj_\kl \kminus \|^2 \\
&= \frac{N}{N+4} \sum_{\theta_\kl = \pi} \frac{2 \dim(\opW_\kl)}{N} \in O \Big( \! \frac{1}{\sqrt{N}} \Big),
\end{align*}
where the final bound follows from the property that there are $O(\sqrt{N})$ terms in the sum.

Thus, using \myeqref{eq:kft} and the property that $\projwfo \kplus$ and $\projwfo \kfo$ are orthogonal, \begin{align*}
g_0^2 &= \frac{3N (N+4)}{4(N+1)^2} \smallspace \| \projwfo \kplus \|^2 + \frac{(2 - N)^2}{4 (N+1)^2} \smallspace \| \projwfo \kfo \|^2 \\
&= \frac{3N (N+4)}{4(N+1)^2} \bigg( \frac{2(N+2)}{3N} \bigg) + O \Big(\! \frac{1}{\sqrt{N}} \Big)\\
&= \frac{1}{2} + O \Big( \! \frac{1}{\sqrt{N}} \Big).
\end{align*}

To bound $g_1^2$, consider the eigenvector decomposition given in \myeqref{eq:evec_def} in the case where $\alpha = \phione$. Then $\ket{A_1^+}$ is the normalized version of $\evec$. By definition, we have $\| \evec \|^2 \geq \frac{1}{2}$, so 
\begin{align*}
g_1^2 &= | \inner{A_1^+}{f_2} |^2 \\
&= \cos^2 (2 \fangle) | \inner{A_1^+}{f_1} |^2  + \sin^2(2 \fangle) | \inner{A_1^+}{+} |^2 \\
&= \cos^2 (2 \fangle) | \inner{A_1^+}{f_1} |^2  \\
&\leq | \inner{\tevec}{f_1} |^2 / \| \evec \|^2 \\
&\leq 2 | \inner{\tevec}{f_1} |^2 \\
&= 2 \Bigg| \tinyspace \aco \tinyspace \sqrt{\frac{N}{N+4}} \bigg(1 - \frac{4}{N} \bigg) - \frac{2}{\sqrt{N+4}} \Bigg|^2 \in O \Big( \frac{1}{\log N} \Big),
\end{align*}
where the final bound follows from Lemma \ref{lem:aco}.
\end{proof}

Lemma \ref{lem:kft_coefficients} shows that $\kft$ has a large constant overlap with the $(+1)$-eigenspace of $\opW \opFo$, and a vanishing overlap with the slowest rotational space. Intuitively, this suggests the reflection $\opFt$ will have little effect on the slowest rotational subspace of the intermediate walk $\opW \opFo$. As discussed in Section~\ref{sec:WFo}, this subspace is where most of the action of the walk takes place. This observation is the basis for the proof of Lemma \ref{lem:WF_proj}.

In the next lemma, we use the constraints from Lemma \ref{lem:flipflop_constraints} to show that the smallest positive eigenphase of $\opW \opF$ is asymptotically close to $\phione$.

\thmsp
\begin{lemma} \label{lem:WF_ephase}
Let $\beta$ denote the smallest eigenphase of $\opW \opF$. Then \begin{align}
\beta = \phione - O \bigg( \frac{1}{\sqrt{\nsize} (\log \nsize)^{3/2}} \bigg).
\end{align}
\end{lemma}
Note that in particular, this implies $\beta \in \Theta ( \! \frac{1}{\sqrt{\nsize \log \nsize}} \! )$.

\begin{proof}
First, we derive an upper bound on $\beta$ using the flip-flop theorem.
By Lemma \ref{lem:claim1}, $\kplus$ intersects every eigenspace of $\opW$. Each of the eigenvectors $\proj_\kl \kplus$ is also an eigenvector of $\opW \opFo$, so in particular, this implies that $g_0 > 0$, $g_{-1} > 0$, and $g_k > 0$ for the $k$ corresponding to the eigenphases $\theta_\kl$. Therefore, by Theorem \ref{lem:flipflop_interlace}, $0 < \beta < \phione$. Applying Lemma \ref{lem:WFo_ephase}, we obtain $\beta \in O(\frac{1}{\sqrt{N \log N}})$.

To obtain the upper bound on $\phione - \beta$, we apply Lemma \ref{lem:flipflop_constraints}, which states that $\beta$ must satisfy \begin{align}
g_0^2 \cot \! \bigg( \frac{\beta}{2} \bigg) + \sum_{k=1}^{m} g_k^2 \Bigg[ \cot \! \bigg( \frac{\varphi_k + \beta}{2}\bigg) - \cot \! \bigg( \frac{\varphi_k - \beta}{2}\bigg) \Bigg] - g_{-1}^2 \tan \! \bigg(\frac{\beta}{2} \bigg) = 0.
\end{align}

We apply trigonometric identities to rewrite this as \begin{align*}
g_0^2 \cot \! \bigg( \frac{\beta}{2} \bigg) - 2 \cot \! \bigg(\frac{\beta}{2} \bigg) \sum_{k=1}^{m} g_k^2 \Bigg( \frac{\cot^2(\frac{\varphi_k}{2}) + 1}{\cot^2(\frac{\beta}{2}) - \cot^2(\frac{\varphi_k}{2})}  \Bigg) - g_{-1}^2 \tan \! \bigg(\frac{\beta}{2} \bigg) = 0.
\end{align*}

Applying our upper bound on $\beta$, we know that $g_{-1}^2 \tan^2 (\frac{\beta}{2}) \in O(\frac{1}{N \log N})$, so we have \begin{align}
2 \sum_{k=1}^m g_k^2 \frac{\cot^2 (\frac{\varphi_k}{2} ) + 1}{\cot^2(\frac{\beta}{2}) - \cot^2 (\frac{\varphi_k}{2}) } =  g_0^2 - O \Big( \frac{1}{N \log N} \Big) = \frac{1}{2} + O \Big( \frac{1}{\sqrt{N}} \Big), \label{eq:beta_bound_pt1}
\end{align}
where the last equality follows by Lemma \ref{lem:kft_coefficients}.

Because the smallest positive eigenphase of $\opW$ has order $\Theta(\frac{1}{\sqrt{N}})$, Theorem \ref{lem:flipflop_interlace} implies that $\varphi_k \in \Omega(\frac{1}{\sqrt{N}})$ for $k \geq 2$. Therefore, $\cot^2 (\frac{\varphi_k}{2}) \in O(N)$ for $k \geq 2$, while $\cot^2(\frac{\beta}{2}) \in \Omega(N \log N)$. Also note that $\sum_k g_k^2 \leq 1$. This means that \begin{align}
 \sum_{k=2}^m g_k^2 \frac{\cot^2 (\frac{\varphi_k}{2} ) + 1}{\cot^2(\frac{\beta}{2}) - \cot^2 (\frac{\varphi_k}{2}) } \in O \Big(\frac{1}{\log N} \Big). \label{eq:beta_bound_pt2}
\end{align}

Combining \myeqref{eq:beta_bound_pt1} and \myeqref{eq:beta_bound_pt2}, we get \begin{align*}
2 g_1^2 \frac{\cot^2 (\frac{\phione}{2} ) + 1}{\cot^2(\frac{\beta}{2}) - \cot^2 (\frac{\phione}{2}) } = \frac{1}{2} - O \Big(\frac{1}{\log N} \Big).
\end{align*}
We know that $\cot^2(\frac{\phione}{2}) \in O(N \log N)$ by Lemma \ref{lem:WFo_ephase}. We also have $g_1^2 \in O(\frac{1}{\log N})$ by Lemma \ref{lem:kft_coefficients}. This implies that \begin{align*}
\cot^2 \! \bigg(\frac{\beta}{2} \bigg) - \cot^2 \! \bigg(\! \frac{\phione}{2}  \! \bigg) \in O(N).
\end{align*}
Because $\cot(\frac{\beta}{2}) + \cot (\frac{\phione}{2}) \in \Omega(\sqrt{N \log N})$, it must be the case that \begin{align*}
\cot \! \bigg(\frac{\beta}{2} \bigg) - \cot \! \bigg( \! \frac{\phione}{2} \! \bigg) \in O \bigg( \! \frac{\sqrt{N}}{\sqrt{\log N}} \! \bigg).
\end{align*}
Applying the Taylor expansion for cotangent, we get \begin{align*}
\frac{1}{\beta} - \frac{1}{\phione} = \frac{\phione - \beta}{\phione \beta} \in O \bigg( \! \frac{\sqrt{N}}{\sqrt{\log N}} \! \bigg). 
\end{align*}
We know that $\frac{1}{\phione} \in \Theta(\sqrt{N \log N})$, so this implies that $\frac{1}{\beta} \in \Theta(\sqrt{N \log N})$. Thus, \begin{align*}
\phione - \beta \in O \bigg( \! \frac{1}{\sqrt{\nsize} (\log \nsize)^{3/2}} \! \bigg).
\end{align*}
\end{proof}

Finally, we show that $\kpiz + \kselfloop$ lies asymptotically in the slowest rotational subspace of $\opW \opF$. To do so, we apply our bounds from Lemmas \ref{lem:kft_coefficients} and \ref{lem:WF_ephase} to show that the slowest eigenvectors of $\opW \opF$ are asymptotically close to the slowest eigenvectors of $\opW \opFo$. 

\thmsp
\begin{lemma} \label{lem:WF_proj}
Let $\proj_\beta$ denote the projection onto the slowest rotational subspace of $\opW \opF$, which is spanned by the eigenvectors with eigenphases $\pm \beta$. Then \begin{align}
\Big \| \proj_\beta \Big(\kpiz + \kselfloop \Big) \Big\| = \sqrt{2} - O \Big(\frac{1}{\log N} \Big).
\end{align}
\end{lemma}

\begin{proof}
We use the decomposition of $\kft$ in \myeqref{eq:kft_decomp}. By Lemma \ref{lem:flipflop_constraints}, the (unnormalized) eigenvector of $\opW \opF$ associated with eigenphase $\beta$ is $\ket{e_\beta} = \kft + \imath \ket{e_\beta^\perp}$, where \begin{equation}
\begin{aligned}
\ket{e_\beta^\perp} = g_0 \cot \! \bigg( \frac{\beta}{2} \bigg) \ket{A_0} + \sum_{k=1}^{m} g_k \Bigg[ \cot \! \bigg( \frac{\beta - \varphi_k}{2} \bigg) \ket{A_k^+} + \cot \! \bigg( \frac{\beta + \varphi_k}{2} \bigg) \ket{A_k^-} \Bigg] \\
- g_{-1} \tan \! \bigg( \frac{\beta}{2} \bigg) \ket{A_{-1}}. \label{eq:e_beta_perp_def}
\end{aligned}
\end{equation}

Let $\ket{B_1^+}$ denote the normalization of $\ket{e_\beta}$, and let $\ket{B_1^-}$ denote the conjugate of $\ket{B_1^+}$. Then $\proj_\beta$ is a projection onto the span of $\ket{B_1^+}$ and $\ket{B_1^-}$.

We know from the proof of Lemma \ref{lem:WFo_proj} that 
\begin{align*}
\Big| \Big( \bra{A_1^+} + \bra{A_1^-} \Big) \Big( \kpiz + \kselfloop \Big) \Big| = 2 - O\Big(\frac{1}{\log N} \Big),
\end{align*}
where $\ket{A_1^+}, \ket{A_1^-}$ are the normalizations of $\evec$ and $\evecb$, respectively. We prove Lemma \ref{lem:WF_proj} by showing that $\ket{A_1^+}$ and $\ket{A_1^-}$ have large overlap with $\ket{B_1^+}$ and $\ket{B_1^-}$, respectively.

We know from \myeqref{eq:e_beta_perp_def} that \begin{align*}
\| \ket{e_\beta} \|^2 &= \| \ket{e_\beta^\perp} \|^2 + \| \kft \|^2 \\
&= g_0^2 \cot^2 \! \bigg(\frac{\beta}{2} \bigg) + \sum_{k=1}^{m} g_k^2 \Bigg[ \cot^2 \! \bigg(\frac{\beta - \varphi_k}{2} \bigg) + \cot^2 \! \bigg( \frac{\beta + \varphi_k}{2} \bigg) \Bigg] + g_{-1}^2 \tan^2 \! \bigg(\frac{\beta}{2} \bigg) + 1.
\end{align*}
By Lemmas \ref{lem:kft_coefficients} and \ref{lem:WF_ephase}, we know that $g_0^2 \cot^2(\frac{\beta}{2}) \in O(\nsize \log \nsize)$ and that $g_{-1}^2 \tan^2(\frac{\beta}{2}) \in O(\frac{1}{\nsize \log \nsize})$. Recall from \myeqref{eq:beta_bound_pt2} that \begin{align*}
\sum_{k=2}^{m} g_k^2 \Bigg[ \cot^2 \! \bigg(\frac{\beta - \varphi_k}{2} \bigg) + \cot^2 \! \bigg(\frac{\beta + \varphi_k}{2} \bigg) \Bigg] \in O \Big(\frac{1}{\log \nsize} \Big).
\end{align*}
Finally, we know from Lemma \ref{lem:WF_ephase} that $\cot^2(\frac{\beta + \phione}{2}) \in O(\nsize \log \nsize)$. Thus, \begin{align*}
\| \ket{e_\beta} \|^2 = g_1^2 \cot^2 \! \bigg( \frac{\beta - \phione}{2} \bigg) + O(\nsize \log \nsize).
\end{align*}

By Lemma \ref{lem:kft_coefficients}, $g_1^2 \in O(\frac{1}{\log \nsize})$, and by Lemma \ref{lem:WF_ephase}, $\cot^2(\frac{\beta - \phione}{2}) \in \Omega( \nsize(\log \nsize)^{3})$. Therefore,
\begin{align*}
|\inner{A_1^+}{B_1^+}|^2 &= \frac{|\inner{A_1^+}{e_\beta}|^2}{\| \ket{e_\beta} \|^2} \\
&= \frac{|\inner{A_1^+}{f_2} + \imath \inner{A_1^+}{e_\beta^\perp}|^2}{\| \ket{e_\beta} \|^2} \\
&= \frac{\Big| g_1 + \imath g_1 \cot \Big(\frac{\beta - \phione}{2} \Big) \Big|^2}{\| \ket{e_\beta} \|^2}\\
&= 1 - O \Big( \frac{1}{\log \nsize} \Big),
\end{align*}
so $|\inner{A_1^+}{B_1^+}| = 1 - O(\frac{1}{\log \nsize})$. Similarly, one can show that \begin{align*}
|\inner{A_1^-}{B_1^+}| &\in O \Big( \frac{1}{\log^2 \nsize} \Big), \\
|\inner{A_1^+}{B_1^-}| &\in O \Big( \frac{1}{\log^2 \nsize} \Big), \\
|\inner{A_1^-}{B_1^-}| &= 1 - O \Big( \frac{1}{\log \nsize} \Big).
\end{align*}
Combining these results, we get \begin{align*}
\Big| \Big( \bra{B_1^+} + \bra{B_1^-} \Big) \Big( \ket{A_1^+} + \ket{A_1^-} \Big) \Big| = 2 - O \Big( \frac{1}{\log \nsize} \Big).
\end{align*}

Therefore, \begin{align*}
\Big \| \proj_\beta \Big( \kpiz + \kselfloop \Big) \Big \| &\geq \frac{1}{\sqrt{2}} \Big| \Big( \bra{B_1^+} + \bra{B_1^-} \Big) \Big( \kpiz + \kselfloop \Big) \Big| \\
&\geq \frac{1}{2\sqrt{2}} \Big| \Big( \bra{B_1^+} + \bra{B_1^-} \Big) \Big( \ket{A_1^+} + \ket{A_1^-} \Big) \Big| \Big| \Big(\bra{A_1^+} + \bra{A_1^-} \Big) \Big( \kpiz + \kselfloop \Big) \Big| \\
& =  \sqrt{2} - O \Big( \frac{1}{\log \nsize} \Big).
\end{align*}
\end{proof}

\subsection{Proof of main result}

Through Lemmas \ref{lem:WF_ephase} and \ref{lem:WF_proj}, we have an asymptotic description of the slowest rotational subspace of $\opW \opF$. The description shows that as $\nsize$ increases, the action of $\opW \opF$ on $\kpiz$ approaches a rotation in this slowest rotational subspace. This property is what allows us to map our initial state to the target state $\kselfloop$ with probability approaching 1. Applying the relationship $\opW \opF = \subz{\opU}$, we thus obtain the proof of our main result.

\thmsp
\mainresult*

\begin{proof}
First, observe that for any $k$, \begin{align*}
\opU^k \kpi = \cz (\opW \opF)^k \cz \kpi = \cz (\opW \opF)^k \kpiz.
\end{align*}
Thus, we prove that after $c \sqrt{\nsize \log \nsize}$ applications of $\opW \opF$ to $\kpiz$, measuring the state will produce $\cz \kselfloop = \kselfloop$ with the stated probability.

Recall from \myeqref{eq:piz_decomp} that $\kpiz$ can be decomposed as \begin{align*}
\kpiz = \frac{1}{2} \kustat + \frac{1}{2} \Big( \kpiz + \kselfloop \Big) .
\end{align*}
By Lemma \ref{lem:ustat}, we know that $\kustat$ is a $(+1)$-eigenvector of $\opW \opF$. Letting $\proj_\beta$ denote the projection onto the slowest rotational subspace of $\opW \opF$, we decompose $\kpiz + \kselfloop$ as \begin{align*}
\kpiz + \kselfloop = \proj_\beta \Big( \kpiz + \kselfloop \Big) + \ket{\bot}.
\end{align*}
for some vector $\ket{\bot}$. By Lemma \ref{lem:WF_proj}, we know that $\| \ket{\bot} \| \in O(\frac{1}{\log \nsize})$. We also know from Lemma \ref{lem:WF_ephase} that applying $\opW \opF$ to a vector in the slowest rotational subspace will result in a rotation of the vector by the angle $\beta \in \Theta(\frac{1}{\sqrt{\nsize \log \nsize}})$. Therefore, there exists a constant $c$ such that $c \sqrt{\nsize \log \nsize} = \lfloor \frac{\pi}{\beta} \rfloor = k$. After $k$ applications of $\opW \opF$ to $\kpiz$, we get the state \begin{align*}
(\opW \opF)^k \kpiz &= \frac{1}{2} \Big( \kpiz - \kselfloop \Big) - \frac{1}{2} \proj_\beta \Big( \kpiz + \kselfloop \Big) + \ket{\rho} \\
&= - \kselfloop + \frac{1}{2} \ket{\bot} + \ket{\rho}.
\end{align*}
Here, $\ket{\rho}$ is some state that captures both the result of applying $(\opW \opF)^k$ to $\ket{\bot}$ and the small error incurred by the rounding of $\frac{\pi}{\beta}$, and has norm $\| \ket{\rho} \| \in O(\frac{1}{\log \nsize})$.
Thus, measuring the state will produce $\kselfloop$ with probability $1 - e(\nsize)$, where $e(\nsize) \in O(\frac{1}{\log \nsize})$.
\end{proof}

\section{Conclusion}

We give a quantum walk that uses minimal memory and $\Theta(\sqrt{\nsize \log \nsize})$ steps to find a unique marked vertex on a two-dimensional grid. In doing so, we show how interpolated walks can be adapted to the memoryless setting. By adding a selfloop to the marked vertex, our walk boosts the probability of measuring the marked state from $O(\frac{1}{\log \nsize})$ to $1 - O(\frac{1}{\log \nsize})$, while preserving the simplicity of the tessellation-based structure. 

We give a precise analysis of how the selfloop affects the walk dynamics by showing that our walk asymptotically reduces to a rotation in a single two-dimensional subspace. Applying this rotation evolves the initial state to the selfloop state, from which the marked state can be obtained by straightforward amplitude amplification.

As part of our proof, we give a precise description of the slowest rotational subspace of our memoryless walk operator. This is done using its decomposition into a real operator composed with a two-dimensional rotation. The techniques we use to analyse such an operator are general enough they have the potential to be used in the analysis of other walks as well. This includes developing and analysing memory-optimal spatial search algorithms for other types of graph, as well as for handling graphs with multiple marked vertices.

\section*{Acknowledgements}

The authors are grateful to Zhan Yu for collaboration in the early stages of the project. This work was supported in part by the Alberta Graduate Excellence Scholarship program (AGES), the Alberta Innovates Graduate Student Scholarships program, and the National Sciences and Engineering Research Council of Canada (NSERC).


\newcommand{\etalchar}[1]{$^{#1}$}


\section{Appendix}

\subsection{Composition with a reflection} \label{app:flipflop}

We discuss techniques to characterize the spectra for the composition of a real operator with a one-dimensional reflection. Operators with this structure appear at multiple points in our work. In this section, we present two lemmas from the quantum walk literature that we apply in our analysis of both $\opW \opFo$ and $\opW \opF$.

Consider an arbitrary real unitary operator $\opT$ acting on a space $\hilb$, and let $\ks\in \hilb$ be a state with real amplitudes. Define $\refs = \opid - \ketbra{\ts}{\ts}$ to be the reflection of state $\ks$. The goal of this section is to describe the spectra of the operator $\opT \refs$.

Because $\opT$ is real-valued, its eigenvalues different from $\pm 1$ come in complex conjugate pairs. We denote the eigenvalues as $e^{\pm \imath \phi_k}$ for $k = 1, 2, \hdots m$, corresponding to the eigenvectors $\ket{T_k^\pm}$. We then decompose $\ks$ into the eigenbasis of $\opT$ as \begin{align}
\ks = s_0 \ket{T_0} + \sum_k s_k \Big( \ket{T_k^+} + \ket{T_k^-} \Big) + s_{-1} \ket{T_{-1}}.
\end{align}
Here, $\ket{T_0}$ and $\ket{T_{-1}}$ are eigenvectors of $\opT$ with eigenvalues $+1$ and $-1$, respectively. The coefficients $s_0$, $s_{-1},$ and all $s_k$ are chosen to be non-negative real numbers by multiplying the eigenvectors with appropriate phases. This decomposition allows us to state the following lemma, originally given by~\cite{Amb07}.

\thmsp
\begin{lemma} \label{lem:flipflop_constraints}
Consider the (unnormalized) state $\ke = \ks + \imath \keperp$, where \begin{align}
\keperp = s_0 \cot \! \bigg( \frac{\alpha}{2} \bigg) \ket{T_0} + \sum_k s_k \Bigg[ \cot \! \bigg( \frac{\alpha - \phi_k}{2} \bigg) \ket{T_k^+} + \cot \! \bigg( \frac{\alpha + \phi_k}{2} \bigg) \ket{T_k^-}  \Bigg] - s_{-1} \tan \! \bigg(\frac{\alpha}{2} \bigg) \ket{T_{-1}},
\end{align}
and $\keperp$ is orthogonal to $\ks$. 
If $\alpha$ is a solution of the equation \begin{align}
s_0^2 \cot \! \bigg( \frac{\alpha}{2} \bigg) + \sum_k s_k^2 \Bigg[ \cot \! \bigg(\frac{\alpha - \phi_k}{2} \bigg) + \cot \! \bigg( \frac{\alpha + \phi_k}{2} \bigg) \Bigg] - s_{-1}^2 \tan \! \bigg(\frac{\alpha}{2} \bigg) = 0,
\end{align}
then $\ke$ is an eigenvector of $\opT \refs$ with eigenvalue $e^{\imath \alpha}$.
\end{lemma}
\thmsp

This lemma allows us to determine the eigenvectors and eigenvalues of $\opT \refs$ by specifying a set of constraints they must satisfy. The lemma is applied in~\cite{Amb07}, and with slight variations in~\cite{AKR05},~\cite{Tul08} and~\cite{DH17b}, to obtain bounds on the smallest eigenphase of a walk operator. We use the lemma for the same purpose, applying it to obtain a lower bound for the smallest eigenphase of $\opW \opFo$ and $\opW \opF$ in Lemmas~\ref{lem:WFo_ephase} and \ref{lem:WF_ephase}. We also use a similar technique in our analysis of the eigenvector $\evec$ in Lemma~\ref{lem:WFo_constraints}, where we derive a set of constraints and use them to find properties of $\aco$ and~$\kpsi$.

The next theorem we state describes the behaviour of the eigenphases of $\opT \refs$ in relation to those of $\opT$. The flip-flop theorem of~\cite{DH17b} describes how the eigenphases of the operators interlace, with the exact pattern of interlacing depending on the eigenspaces of $\opT$ that $\ks$ intersects. We limit the theorem statement to the case we apply in this paper, where $\ks$ intersects the $(+1)$-eigenspace, the $(-1)$-eigenspace, and at least one other eigenspace of $\opT$.

\thmsp
\begin{theorem}[Flip-flop theorem] \label{lem:flipflop_interlace}
Consider any real unitary $\opT$ and let $\ks$ be a state with real amplitudes in the same space. Denote the positive eigenphases of $\opT$ different from $0, \pi$ by $0 < \phi_1 \leq \phi_2 \leq \cdots \leq \phi_m < \pi$. If $s_0 \neq 0$, $s_{-1} \neq 0$ and $s_k \neq 0$ for some $k$, then $\opT \refs$ has $m+1$ two-dimensional eigenspaces, and no $(+1)$- or $(-1)$-eigenspaces which overlap $\ks$. The positive eigenphases $\alpha_j$ of $\opT \refs$ satisfy the inequality $0 < \alpha_0 < \phi_1 \leq \alpha_1 \leq \cdots \leq \phi_m \leq \alpha_m < \pi$.
\end{theorem}
\thmsp

We apply this theorem in Lemmas~\ref{lem:WFo_ephase} and \ref{lem:WF_ephase} to obtain an upper bound on the smallest positive eigenphases of $\opW \opFo$ and $\opW \opF$, respectively. One of the contributions of our work is to show how Theorem \ref{lem:flipflop_interlace} can be used in combination with Lemma \ref{lem:flipflop_constraints} to tightly bound these eigenphases. We show that this approach can be used in the case of an operator composed with a reflection, and then by applying a second reflection, to an operator composed with a two-dimensional rotation.

\subsection{Decomposition of $\kplus$ and $\kminus$} \label{app:proof_of_claim1}

To analyse the behaviour of $\opW \opFo$ and $\opW \opF$, we specify how $\opW$ acts on vectors in the non-trivial eigenspaces of $\opF$. Recall from the definitions in \myeqref{eq:kplus_def} and \myeqref{eq:kminus_def} that $\kplus$ and $\kminus$ are orthonormal vectors that have the same span as $\kmarked$ and $\keg$. Together with $\kselfloop$, they span a space that includes the two-dimensional subspace on which $\opF$ acts non-trivially. In this appendix, we prove Lemma \ref{lem:claim1}, which describes how $\kplus$ and $\kminus$ decompose into the invariant subspaces of $\opW$.

To simplify notation, define
\begin{align*}
  { \everymath={\displaystyle}
  \begin{array}{l@{\;=\;}c@{\;=\;}l@{\;=\;}r}
    s_\kl^+ 
    & \frac{1}{2}(r_\kl^+ + r_\kl^-)
    & \multicolumn{2}{@{}l}{
      \hphantom{\signl} \sqrt{1+ \frac{\sin \tl}{p_\kl}}} \\[4.8mm]
    s_\kl^- 
    & \frac{1}{2}(r_\kl^+ - r_\kl^-)
    & \multicolumn{2}{@{}l}{
      \signl \sqrt{1- \frac{\sin \tl}{p_\kl}}}
  \end{array}
      }
\end{align*}
and
\begin{align*}
  { \everymath={\displaystyle}
  \begin{array}{l@{\;=\;}c@{\;=\;}l@{\;=\;}r}
    d_\kl^+ 
    & \frac{1}{2}(c_\kl^+ + c_\kl^-)
    & \multicolumn{2}{@{}l}{
      \hphantom{\signk} \sqrt{1+ \frac{\sin \tk}{p_\kl}}} \\[4.8mm]
    d_\kl^- 
    & \frac{1}{2}(c_\kl^+ - c_\kl^-)
    & \multicolumn{2}{@{}l}{
      \signk \sqrt{1- \frac{\sin \tk}{p_\kl}}}.
  \end{array}
      }
\end{align*}

In the case where $k=l=0$, we define $s^+_{00} = d^+_{00} = \sqrt{2}$ and $s^-_{00} = d^-_{00} = 0$. Note that $r_\kl^\pm = s_\kl^+ \pm s_\kl^-$ and $c_\kl^\pm = d_\kl^+ \pm d_\kl^-$.

Recall that both $\keg$ and $\kmarked$ lie in the span of the basis states $\ket{00}$, $\ket{01}$, $\ket{10}$ and $\ket{11}$.
We compute the projections of these basis states
onto the components of the eigenvectors of~$\opW$.

\begin{minipage}{.5\textwidth}
  \begin{align}
    \inner{0}{u_\kl}
    &= \sqrt{2} \inner{0}{r_\kl} \inner{0}{\phi^k_r}
      = \frac{1}{\sqrt{2 \nrows}} r_\kl^{-} \notag \\
    \inner{1}{u_\kl}
    &= \sqrt{2} \inner{1}{r_\kl} \inner{1}{\phi^k_r}
      = \frac{1}{\sqrt{2 \nrows}} r_\kl^{+} \omega^k_{\nrows} \notag
  \end{align}
\end{minipage}
\begin{minipage}{.5\textwidth}
  \begin{align}
    \inner{0}{v_\kl}
    &= \sqrt{2} \inner{0}{c_\kl} \inner{0}{\phi^l_c}
      = \frac{1}{\sqrt{2 \ncols}} c_\kl^{-} \notag\\
    \inner{1}{v_\kl}
    &= \sqrt{2} \inner{1}{c_\kl} \inner{1}{\phi^l_c}
      = \frac{1}{\sqrt{2 \ncols}} c_\kl^{+} \omega^l_{\ncols} \notag
  \end{align}
\end{minipage}

\begin{minipage}{.5\textwidth}
  \begin{align*}
    \inner{0}{u^1_\kl}
    &= \sqrt{2} \inner{0}{r^1_\kl} \inner{0}{\phi^k_r}
      = - \frac{1}{\sqrt{2 \nrows}} r_\kl^{+} \\
    \inner{1}{u^1_\kl}
    &= \sqrt{2} \inner{1}{r^1_\kl} \inner{1}{\phi^k_r}
      = \frac{1}{\sqrt{2 \nrows}} r_\kl^{-} \omega^k_{\nrows} 
  \end{align*}
\end{minipage}
\begin{minipage}{.5\textwidth}
  \begin{align*}
    \inner{0}{v^1_\kl}
    &= \sqrt{2} \inner{0}{c^1_\kl} \inner{0}{\phi^l_c}
      = - \frac{1}{\sqrt{2 \ncols}} c_\kl^{+} \\
    \inner{1}{v^1_\kl}
    &= \sqrt{2} \inner{1}{c^1_\kl} \inner{1}{\phi^l_c}
      = \frac{1}{\sqrt{2 \ncols}} c_\kl^{-} \omega^l_{\ncols}.
  \end{align*}
\end{minipage}

Now, using the property that
\begin{align*}
  \frac{1}{2}
  \Big(r_\kl^{+} \omega^k_{\nrows} + r_\kl^{-}\Big) 
  &= \frac{1}{2} \omega^{k/2}_{\nrows}
    \Big(
    r_\kl^{+} \omega^{k/2}_{\nrows}
    + r_\kl^{-}\omega_\nrows^{-k/2} 
    \Big) 
    = \omega^{k/2}_{\nrows}
    \Big(
    \cos\bigg( \frac{\tk}{2} \bigg) s_\kl^{+}
    + \imath \sin \bigg( \frac{\tk}{2}\bigg) s_\kl^{-}
    \Big)    \\
  \frac{1}{2}
  \Big(r_\kl^{-} \omega^k_{\nrows} - r_\kl^{+}\Big) 
  &= \frac{1}{2} \omega^{k/2}_{\nrows}
    \Big(
    r_\kl^{-} \omega^{k/2}_{\nrows}
    - r_\kl^{+}\omega_\nrows^{-k/2} 
    \Big) 
    = \omega^{k/2}_{\nrows}
    \Big(
    -\cos \bigg(\frac{\tk}{2} \bigg) s_\kl^{-}
    + \imath \sin\bigg( \frac{\tk}{2} \bigg) s_\kl^{+}
    \Big),    
\end{align*}
we compute
\begin{align*}
  \frac{1}{\sqrt{2}} (\bra{0} + \bra{1}) \ket{u_\kl}
  &= \frac{1}{\sqrt{\nrows}} \omega^{k/2}_{\nrows}
    \Big(
    \cos \bigg( \frac{\tk}{2} \bigg) s_\kl^{+}
    + \imath \sin\bigg(\frac{\tk}{2} \bigg) s_\kl^{-}
    \Big) \\
  \frac{1}{\sqrt{2}} (\bra{0} + \bra{1}) \ket{u^1_\kl}
  &= \frac{1}{\sqrt{\nrows}} \omega^{k/2}_{\nrows}
    \Big(
    -\cos\bigg(\frac{\tk}{2}\bigg) s_\kl^{-}
    + \imath \sin\bigg(\frac{\tk}{2}\bigg) s_\kl^{+}
    \Big) \\
  \frac{1}{\sqrt{2}} (\bra{0} + \bra{1}) \ket{v_\kl}
  &= \frac{1}{\sqrt{\ncols}} \omega^{l/2}_{\ncols}
    \Big(
    \cos\bigg(\frac{\tl}{2}\bigg) d_\kl^{+}
    + \imath \sin\bigg(\frac{\tl}{2}\bigg) d_\kl^{-}
    \Big) \\
  \frac{1}{\sqrt{2}} (\bra{0} + \bra{1}) \ket{v^1_\kl}
  &= \frac{1}{\sqrt{\ncols}} \omega^{l/2}_{\ncols}
    \Big(
    -\cos\bigg(\frac{\tl}{2}\bigg) d_\kl^{-}
    + \imath \sin\bigg(\frac{\tl}{2}\bigg) d_\kl^{+}
    \Big).
\end{align*}

Excluding the case $k=l=0$, the squared amplitudes of the projections are then 
\newline
\begin{minipage}{.49\textwidth}
  \begin{align*}
    \big\| \inner{0}{u_\kl} \big\|^2
    &= \frac{1}{\nrows} 
      \bigg(
      1 - \frac{\sin \tk  \cos\tl}{p_\kl}\,
      \bigg) \\
    \big\| \inner{0}{u^1_\kl} \big\|^2
    &= \frac{1}{\nrows} 
      \bigg(
      1 + \frac{\sin \tk  \cos\tl}{p_\kl}\,
      \bigg) \\
    \big\| \inner{0}{v_\kl} \big\|^2
    &= \frac{1}{\ncols} 
      \bigg(
      1 - \frac{\cos \tk  \sin\tl}{p_\kl}\,
      \bigg)\\
    \big\| \inner{0}{v^1_\kl} \big\|^2
    &= \frac{1}{\ncols} 
      \bigg(
      1 + \frac{\cos \tk  \sin\tl}{p_\kl}\,
      \bigg)\\
  \end{align*}
\end{minipage}
\begin{minipage}{.49\textwidth}
  \begin{align*}
    \bigg\| \frac{1}{\sqrt{2}} (\bra{0} + \bra{1}) \ket{u_\kl} \bigg\|^2
    &= \frac{1}{\nrows} 
      \bigg(
      1 + \frac{\cos \tk \sin \tl}{p_\kl}\,
      \bigg) \\
    \bigg\| \frac{1}{\sqrt{2}} (\bra{0} + \bra{1}) \ket{u^1_\kl} \bigg\|^2
    &= \frac{1}{\nrows} 
      \bigg(
      1 - \frac{\cos \tk \sin \tl}{p_\kl}\,
      \bigg) \\ 
    \bigg\| \frac{1}{\sqrt{2}} (\bra{0} + \bra{1}) \ket{v_\kl} \bigg\|^2
    &= \frac{1}{\ncols} 
      \bigg(
      1 + \frac{\sin \tk \cos \tl}{p_\kl}\,
      \bigg)\\
    \bigg\| \frac{1}{\sqrt{2}} (\bra{0} + \bra{1}) \ket{v^1_\kl} \bigg\|^2
    &= \frac{1}{\ncols} 
      \bigg(
      1 - \frac{\sin \tk \cos \tl}{p_\kl}\,
      \bigg).
  \end{align*}
  \vspace*{2mm}
\end{minipage}

\begin{fact} \label{lem:keg_kmarked_proj}
For any subspace $\opW_{kl}$, \begin{align}
\bra{\marked} \proj_{kl} \keg = \begin{cases}
\frac{1}{2} & \quad \text{if $k = l = 0$} \\
0 & \quad \text{otherwise}
\end{cases}.
\end{align}
\end{fact}

\begin{proof}
Recall that $\kmarked = \ket{00}$ and $\keg = \frac{1}{2} (\ket{0} + \ket{1}) \otimes (\ket{0} + \ket{1})$. Consider any $k,l$ not both $0$. Then \begin{align*}
& \quad \inner{\marked}{w_{\kpl}^{11}} \inner{w_{\kpl}^{11}}{\eg} \\
&= \frac{1}{2N} (r_{\kpl}^+ c_{\kpl}^+ ) \Big( \omega_{\nrows}^{\kp /2} \omega_{\ncols}^{l/2} \Big) \Big(-\cos \bigg( \frac{\tkp}{2} \bigg) s_{\kpl}^- + \imath \sin \bigg( \frac{\tkp}{2} \bigg) s_{\kpl}^+ \Big) \Big( -\cos \bigg( \frac{\tl}{2} \bigg) d_{\kpl}^- + \imath \sin \bigg( \frac{\tl}{2} \bigg) d_{\kpl}^+ \Big) \\
&= \frac{-1}{2N} (r_{\kl}^+ c_{\kl}^- ) \Big( \omega_{\nrows}^{-k/2} \omega_{\ncols}^{l/2} \Big) \Big(\cos \bigg( \frac{\tk}{2} \bigg) s_{\kl}^+ + \imath \sin \bigg( \frac{\tk}{2} \bigg) s_{\kl}^- \Big) \Big( \cos \bigg( \frac{\tl}{2} \bigg) d_{\kl}^- + \imath \sin \bigg( \frac{\tl}{2} \bigg) d_{\kl}^+ \Big) \\
&=  \frac{-1}{2N} (r_{\kl}^- c_{\kl}^- ) \Big( \omega_{\nrows}^{k/2} \omega_{\ncols}^{l/2} \Big) \Big(\cos \bigg( \frac{\tk}{2} \bigg) s_{\kl}^+ + \imath \sin \bigg( \frac{\tk}{2} \bigg) s_{\kl}^- \Big) \Big( \cos \bigg( \frac{\tl}{2} \bigg) d_{\kl}^+ + \imath \sin \bigg( \frac{\tl}{2} \bigg) d_{\kl}^- \Big) \\
&= \quad -\inner{\marked}{w_{\kl}^{00}} \inner{w_{\kl}^{00}}{\eg}.
\end{align*}

Similarly, it can be derived that \begin{align*}
\inner{\marked}{w_{\klp}^{11}} \inner{w_{\klp}^{11}}{\eg} &= -\inner{\marked}{w_{\kl}^{00}} \inner{w_{\kl}^{00}}{\eg}.
\end{align*}

Using the property that $(k')' = k$, this further implies that $\inner{\marked}{w_{\kplp}^{00}} \inner{w_{\kplp}^{00}}{\eg} = \inner{\marked}{w_{\kl}^{00}} \inner{w_{\kl}^{00}}{\eg}$. 
By definition of the invariant subspaces $\opW_\kl$ in Section~\ref{sec:invariant_subspaces}, this shows that for any $k, l$ not both $0$, $\bra{\marked} \proj_{kl} \keg = 0$. 
It follows that $\bra{\marked} \proj_{00} \keg = \inner{\marked}{\eg} = \frac{1}{2}$.
\end{proof}

\thmsp
\claimone*

\begin{proof}
Observe that for any $k, l$ not both zero, 
\begin{align*}
& \quad \| \inner{\marked}{w_\kl^{00}} \|^2 + \| \inner{\marked}{w_\klp^{11}} \|^2 + \| \inner{\marked}{w_\kpl^{11}} \|^2 + \| \inner{\marked}{w_\kplp^{00}} \|^2 \\[4mm]
&= \frac{1}{N} \bigg( 1 - \frac{\sin\tk\cos\tl}{p_\kl} \bigg) \! \bigg( 1 - \frac{\cos\tk\sin\tl}{p_\kl} \bigg) + \frac{1}{N} \bigg( 1 - \frac{\sin\tk\cos\tl}{p_\kl} \bigg) \! \bigg( 1 + \frac{\cos\tk\sin\tl}{p_\kl} \bigg) \nonumber \\[4mm]
&\quad + \frac{1}{N} \bigg( 1 + \frac{\sin\tk\cos\tl}{p_\kl} \bigg) \! \bigg( 1 - \frac{\cos\tk\sin\tl}{p_\kl} \bigg) + \frac{1}{N} \bigg( 1 + \frac{\sin\tk\cos\tl}{p_\kl} \bigg) \! \bigg( 1 + \frac{\cos\tk\sin\tl}{p_\kl} \bigg) \\[4mm]
&= \frac{4}{N} = \| \inner{\eg}{w_\kl^{00}} \|^2 + \| \inner{\eg}{w_\klp^{11}} \|^2 + \| \inner{\eg}{w_\kpl^{11}} \|^2 + \| \inner{\eg}{w_\kplp^{00}} \|^2.
\end{align*}

Therefore, for any subspace $\opW_{kl}$ with $kl \neq 00$, we have $\bra{\marked} \proj_{kl} \ket{\marked} = \bra{\eg} \proj_{kl} \ket{\eg} = \frac{\dim(\opW_{kl})}{N}$. Next, recall that $\opW_{00}$ refers to the $(+1)$-eigenspace of $\opW$. We know that both $\kmarked$ and $\keg$ are normalized, so
\begin{align*}
\bra{\marked} \proj_{00} \ket{\marked} = \bra{\eg} \proj_{00} \ket{\eg} = 1 - \sum_{\kl \neq 00} \frac{\dim(\opW_\kl)}{N} = \frac{N+4}{2N}.
\end{align*}

Applying Fact \ref{lem:keg_kmarked_proj}, this implies that
\begin{align*}
\sqrt{3} \bra{-} \proj_{kl} \ket{+} = \bra{\marked} \proj_{kl} \ket{\marked} + \bra{\marked} \proj_{kl} \keg - \bra{\eg} \proj_{kl} \kmarked - \bra{\eg} \proj_{kl} \keg = 0,
\end{align*}
for any subspace $\opW_{kl}$. This proves \myeqref{eq:kp_perp_km}.

To prove \myeqref{eq:kplusj}, we compute
\begin{align*}
\bra{+} \proj_{kl} \ket{+} = \frac{1}{3} \Big[ \bra{\marked} \proj_{kl} \kmarked - \bra{\marked} \proj_{kl} \keg - \bra{\eg} \proj_{kl} \kmarked + \bra{\eg} \proj_{kl} \keg \Big] = \frac{2 \dim(\opW_{kl})}{3N},
\end{align*}
and similarly for \myeqref{eq:kminusj}.

Using the property that $\kplus$ and $\kminus$ are normalized, this implies that \begin{align*}
\| \proj_{00} \kplus \|^2 = 1 - \sum_{kl \neq 00} \frac{2\dim(\opW_{kl})}{3N} = \frac{2(N+2)}{3N},
\end{align*}
which proves \myeqref{eq:kpluszero}. We can similarly compute that $\| \proj_{00} \kminus \|^2 = \frac{4}{N}$, proving \myeqref{eq:kminuszero}.
\end{proof}

\pagebreak
\subsection{Sums}

In this section, we prove asymptotic bounds on a set of sums over the spectra of $\opW$. We assume a square grid, with $\nrows = \ncols = \sqrt{\nsize}$.

\thmsp
\begin{fact} \label{lem:sum0}
Suppose $0 < \alpha < \theta_\kl$ for all $k,l$, and consider the sum
\begin{align}
\sum_{kl \neq 00} \dim(\opW_\kl) \cot \! \bigg(\frac{\theta_\kl - \alpha}{2} \bigg).
\end{align}
\begin{enumerate}
\item If $\alpha \in \Theta(\frac{1}{\sqrt{\nsize}})$, then the sum has order $\Omega(\sqrt{\nsize} \log \nsize)$.

\item If $\alpha \in o(\frac{1}{\sqrt{\nsize}})$, then the sum has order $\Theta(\alpha \nsize \log \nsize)$.
\end{enumerate}
\end{fact}

\begin{proof}
Instead of taking the sum over the subspaces $\opW_\kl$, which partition the domain of $\opW$, we convert to a sum over $k$ and $l$. Recall that each pair $0 \leq k, l \leq \sqrt{\nsize}/2 - 1$ corresponds to two eigenvectors of $\opW$: $\ket{w_\kl^{00}}$ with eigenphase $\theta_\kl$ and $\ket{w_\kl^{11}}$ with eigenphase $-\theta_\kl$. Using this property, we rewrite the sum as 
\begin{equation}
\begin{aligned}[b]
&\sum_{kl \neq 00} \dim(\opW_\kl) \cot \! \bigg(\frac{\theta_\kl - \alpha}{2} \bigg)\\
&= \sum_{k = 0}^{\frac{\sqrt{\nsize}}{2} - 1} \sum_{\substack{l = 0\\ \text{not both 0}}}^{\frac{\sqrt{\nsize}}{2} - 1}  \cot \! \bigg( \frac{\theta_\kl - \alpha}{2} \bigg) - \cot \! \bigg( \frac{\theta_\kl + \alpha}{2} \bigg) \\
&= 2 \sum_{k = 0}^{\frac{\sqrt{\nsize}}{2} - 1} \sum_{\substack{l = 0\\ \text{not both 0}}}^{\frac{\sqrt{\nsize}}{2} - 1} \frac{\cot(\frac{\alpha}{2}) \Big( \cot^2(\frac{\theta_\kl}{2}) + 1 \Big)}{\cot^2(\frac{\alpha}{2}) - \cot^2(\frac{\theta_\kl}{2})}, \label{eq:double_sum}
\end{aligned}
\end{equation}
where the final equality follows from angle sum identities.

Next, observe that by the definition of $\theta_\kl$, \begin{align*}
\cot^2 \! \bigg(\frac{\theta_\kl}{2} \bigg) + 1 = \frac{2}{1 - \cos \theta_\kl} = \frac{1}{1 - \cos^2 \tk \cos^2 \tl}.
\end{align*}
We therefore consider the sum \begin{align*}
\sum_{k = 0}^{\frac{\sqrt{\nsize}}{2} - 1} \sum_{\substack{l = 0\\ \text{not both 0}}}^{\frac{\sqrt{\nsize}}{2} - 1} \frac{1}{1 - \cos^2 \tk \cos^2 \tl}.
\end{align*}
For the terms where $l = 0$, we get \begin{align*}
\sum_{k = 1}^{\frac{\sqrt{\nsize}}{2} - 1} \frac{1}{1 - \cos^2 \tk} \in \Theta(\nsize),
\end{align*}
and similarly for $k = 0$. The remaining terms satisfy 
\begin{align*}
\sum_{k = 1}^{\frac{\sqrt{\nsize}}{2} - 1} \sum_{l = 1}^{\frac{\sqrt{\nsize}}{2} - 1} \frac{1}{1 - \cos^2 \tk \cos^2 \tl} \in \Theta(\nsize \log \nsize).
\end{align*}

Now, consider \myeqref{eq:double_sum} in the case where $\alpha \in \Theta(\frac{1}{\sqrt{\nsize}})$. We have 
\begin{align*}
& 2 \sum_{k = 0}^{\frac{\sqrt{\nsize}}{2} - 1} \sum_{\substack{l = 0\\ \text{not both 0}}}^{\frac{\sqrt{\nsize}}{2} - 1} \frac{\cot(\frac{\alpha}{2}) \Big( \cot^2(\frac{\theta_\kl}{2}) + 1 \Big)}{\cot^2(\frac{\alpha}{2}) - \cot^2(\frac{\theta_\kl}{2})} \\
&\geq 2 \sum_{k = 0}^{\frac{\sqrt{\nsize}}{2} - 1} \sum_{\substack{l = 0\\ \text{not both 0}}}^{\frac{\sqrt{\nsize}}{2} - 1} \frac{\cot(\frac{\alpha}{2}) \Big( \cot^2(\frac{\theta_\kl}{2}) + 1 \Big)}{\cot^2(\frac{\alpha}{2})} \\
& = \frac{2}{\cot(\frac{\alpha}{2})} \sum_{k = 0}^{\frac{\sqrt{\nsize}}{2} - 1} \sum_{\substack{l = 0\\ \text{not both 0}}}^{\frac{\sqrt{\nsize}}{2} - 1} \frac{1}{1 - \cos^2 \tk \cos^2 \tl}.
\end{align*}
Therefore, \begin{align*}
\sum_{kl \neq 00} \dim(\opW_\kl) \cot \! \bigg(\frac{\theta_\kl - \alpha}{2} \bigg) \in \Omega(\sqrt{\nsize} \log \nsize).
\end{align*}

In the case where $\alpha \in o(\frac{1}{\sqrt{\nsize}})$, the denominator in \myeqref{eq:double_sum} is dominated by the term $\cot^2(\frac{\alpha}{2})$. Therefore, the expression has the same asymptotic order as \begin{align*}
\frac{1}{\cot(\frac{\alpha}{2})} \sum_{k = 0}^{\frac{\sqrt{\nsize}}{2} - 1} \sum_{\substack{l = 0\\ \text{not both 0}}}^{\frac{\sqrt{\nsize}}{2} - 1} \frac{1}{1 - \cos^2 \tk \cos^2 \tl} \in \Theta(\alpha \nsize \log \nsize),
\end{align*}
proving the second clause.
\end{proof}

\thmsp
\begin{fact} \label{lem:sum4}
Suppose $0 < \alpha < \theta_\kl$ for all $k,l$, and that $\alpha \in o(\frac{1}{\sqrt{\nsize}})$. Then
\begin{align}
\sum_{\kl \neq 00} \dim(\opW_\kl) \Bigg[ \cot \! \bigg(\frac{\theta_\kl + \alpha}{2} \bigg) - \cot \! \bigg(\frac{\theta_\kl - \alpha}{2} \bigg) \Bigg]^2 \in \Theta(\alpha^2 \nsize^2).
\end{align}
\end{fact}

\begin{proof}
\begin{align*}
&\sum_{\kl \neq 00} \dim(\opW_\kl) \Bigg[ \cot \! \bigg(\frac{\theta_\kl + \alpha}{2} \bigg) - \cot \! \bigg(\frac{\theta_\kl - \alpha}{2} \bigg) \Bigg]^2 \\
&= \sum_{k = 0}^{\frac{\sqrt{\nsize}}{2} - 1} \sum_{\substack{l = 0\\ \text{not both 0}}}^{\frac{\sqrt{\nsize}}{2} - 1} \Bigg[ \cot \! \bigg(\frac{\theta_\kl + \alpha}{2} \bigg) - \cot \! \bigg(\frac{\theta_\kl - \alpha}{2} \bigg) \Bigg]^2
\end{align*}
By a similar derivation as in Fact~(\ref{lem:sum0}), this sum has the same order as \begin{align*}
\frac{1}{\cot^2(\frac{\alpha}{2})} \sum_{k = 0}^{\frac{\sqrt{\nsize}}{2} - 1} \sum_{\substack{l = 0\\ \text{not both 0}}}^{\frac{\sqrt{\nsize}}{2} - 1} \! \bigg( \frac{1}{1 - \cos^2 \tk \cos^2 \tl}\bigg)^2.
\end{align*}
Using the Taylor expansion of cosine, for $l = 0$ we get \begin{align*}
\sum_{k = 0}^{\frac{\sqrt{\nsize}}{2} - 1} \! \bigg( \frac{1}{1 - \cos^2 \tk} \bigg)^2 \in \Theta(\nsize^2),
\end{align*}
and similarly for $k = 0$.
Finally, 
\begin{align*}
\sum_{k = 1}^{\frac{\sqrt{\nsize}}{2} - 1} \sum_{l = 1}^{\frac{\sqrt{\nsize}}{2} - 1} \! \bigg( \frac{1}{1 - \cos^2 \tk \cos^2 \tl} \bigg)^2 \in \Theta(\nsize^2).
\end{align*}
Therefore, \begin{align*}
\sum_{\kl \neq 00} \dim(\opW_\kl) \Bigg[ \cot \! \bigg(\frac{\theta_\kl + \alpha}{2} \bigg) - \cot \! \bigg(\frac{\theta_\kl - \alpha}{2} \bigg) \Bigg]^2 \in \Theta(\alpha^2 \nsize^2).
\end{align*}
\end{proof}

\end{document}